\pdfoutput=1
\RequirePackage{ifpdf}
\ifpdf 
\documentclass[pdftex]{sigma}
\else
\documentclass{sigma}
\fi

\newcommand{\inversion}[1]{\!\overleftrightarrow{\,#1\,}\!}

\newcommand{\reverse}{{\,\overleftrightarrow{\!(\cdots)\!}\,}}

\usepackage{tikz}

\newcommand{\triang}{{\begin{tikzpicture}[scale=.15]
			\draw (1,0) -- (0,0) -- (0,1) -- cycle;
\end{tikzpicture}}}

\newcommand{\trileg}{{\begin{tikzpicture}[scale=.15]
			\draw (1,1) -- (0,0);
			\draw (1,0) -- (0,0) -- (0,1);
\end{tikzpicture}}}

\newcommand{\cyclic}{{\circlearrowleft}_{ijk}}

\newcommand{\pdv}[2]{\frac{\partial {#1}}{\partial {#2}}}

\newcommand{\C}{\mathbb{C}}

\newcommand{\R}{\mathbb{R}}
\newcommand{\Z}{\mathbb{Z}}

\newcommand{\dilog}{\operatorname{Li_2}}

\newcommand{\cL}{\mathcal{L}}

\renewcommand{\d}{\mathrm{d}}

\renewcommand{\vec}{\boldsymbol}

{ \theoremstyle{definition}
	\newtheorem{Definition}{Definition}[section]
	
	\newtheorem{ExampleH1}{Example: H1, part}
	\newtheorem{ExampleH2}{Example: H2, part}
	\newtheorem{ExampleH3}{Example: H3$\boldsymbol{_{\delta=0}}$, part}
	\newtheorem{Remark}[Definition]{Remark}
}
\newtheorem{Lemma}[Definition]{Lemma}
\newtheorem{Theorem}[Definition]{Theorem}
\newtheorem{Proposition}[Definition]{Proposition}

\numberwithin{equation}{section}

\begin{document}
\allowdisplaybreaks

\newcommand{\arXivNumber}{2501.13012}

\renewcommand{\PaperNumber}{058}

\FirstPageHeading

\ShortArticleName{Discrete Lagrangian Multiforms for ABS Equations I}

\ArticleName{Discrete Lagrangian Multiforms for ABS Equations I: Quad Equations}
	
\Author{Jacob J. RICHARDSON~$^{\rm a}$ and Mats VERMEEREN~$^{\rm b}$}

\AuthorNameForHeading{J.J.~Richardson and M.~Vermeeren}

\Address{$^{\rm a)}$~School of Mathematics, University of Leeds, Leeds, LS2 9JT, UK}
\EmailD{\href{mailto:jacobjoseph.gu@gmail.com}{jacobjoseph.gu@gmail.com}}

\Address{$^{\rm b)}$~Department of Mathematical Sciences, Loughborough University, \\
\hphantom{$^{\rm b)}$}~Loughborough, LE11 3TU, UK}
\EmailD{\href{mailto:m.vermeeren@lboro.ac.uk}{m.vermeeren@lboro.ac.uk}}
	
\ArticleDates{Received March 26, 2024, in final form July 02, 2025; Published online July 18, 2025}
	
\Abstract{Discrete Lagrangian multiform theory is a variational perspective on lattice equations that are integrable in the sense of multidimensional consistency. The Lagrangian multiforms for the equations of the ABS classification formed the start of this theory, but the Lagrangian multiforms that are usually considered in this context produce equations that are slightly weaker than the ABS equations. In this work, we present alternative Lagrangian multiforms that have Euler--Lagrange equations equivalent to the ABS equations. In addition, the treatment of the ABS Lagrangian multiforms in the existing literature fails to acknowledge that the complex functions in their definitions have branch cuts. The choice of branch affects both the existence of an additive three-leg form for the ABS equations and the closure property of the Lagrangian multiforms. We give counterexamples for both these properties, but we recover them by including integer-valued fields, related to the branch choices, in the action sums.}
	
	
\Classification{39A36; 37J70; 37J06}

\section{Introduction}

The theory of Lagrangian multiforms describes integrable systems through a variational principle. In recent years, the continuous version of the theory, describing integrable differential equations, has received a lot of attention \cite{caudrelier2023lagrangian,caudrelier2024classical, nijhoff2023lagrangian, petrera2017variational, sleigh2021lagrangian, suris2016lagrangian} and a semi-discrete version was formulated, extending the theory to differential-difference equations \cite{nijhoff2024lagrangianb, sleigh2022semi}. However, the origins of Lagrangian multiform theory lie in the fully discrete setting and concern integrable \emph{quad equations}, in particular those of the Adler--Bobenko--Suris (ABS) list \cite{adler2003classification}. This is a classification of multi-affine, multidimensionally consistent difference equations on a quadrilateral stencil (under the additional assumption that there exists a compatible \emph{tetrahedron equation}). The observation that all equations on this list have a variational structure led to the introduction of Lagrangian multiforms in \cite{lobb2009lagrangian}.

Lagrangian multiform theory describes a set of compatible equations (multidimensionally consistent difference equations or a hierarchy of differential equations) through a single variational principle involving a difference form or differential form. This $d$-form is defined on the space of all independent variables of the set of equations. For any $d$-dimensional surface within the space of independent variables, we can consider an action, defined by the integral/sum of the Lagrangian $d$-form over the surface. In Lagrangian multiform theory (and in the closely related theory of \emph{pluri-Lagrangian systems}), the variational principle requires these integrals to be critical no matter which surface is chosen.

The action functional originally suggested for the ABS equations in \cite{adler2003classification} involves Lagrangians defined on a 4-point stencil. These have not been studied as Lagrangian 2-forms in a higher-dimensional lattice, only as traditional Lagrangians in a planar lattice.
In \cite{lobb2009lagrangian}, Lagrangians on triangular stencils were introduced and generalised to discrete Lagrangian 2-forms.
These have been studied in \cite{bobenko2010lagrangian, boll2014integrability, boll2016integrability, lobb2009lagrangian, lobb2018variational, xenitidis2011lagrangian} and produce a slightly weaker set of equations: the quad equations are sufficient conditions for the variational principle, but not necessary conditions.
Each of the generalised Euler--Lagrange equations (called \emph{corner equations}) is a linear combination of two quad equations \cite{boll2014integrability}, or they can be put in the form of quad equations with an additional term, which plays the role of an integration constant \cite{lobb2018variational}.
It is emphasised in \cite{boll2014integrability} that: \emph{``quad-equations are not variational; rather, discrete Euler--Lagrange equations for the $2$-forms given in {\rm \cite{bobenko2010lagrangian,lobb2009lagrangian}} are consequences of quad-equations''}.
In the present work, we show that the literal interpretation of the first part of that statement is not true: quad equations \emph{are} variational.
We generalise the original Lagrangians on quadrilateral stencils to Lagrangian 2-forms and show that their corner equations imply the multi-affine quad equations.\looseness=1

To obtain equivalence between the corner equations and the multi-affine quad equations, we need to overcome an additional subtlety that has often been ignored in the literature (a notable exception is \cite{bobenko2012discrete}). The corner equations are in an \emph{additive three-leg form}. For some of the ABS equations, transforming this three-leg form into the multi-affine quad equation involves taking its exponential. Hence, the inverse operation is only defined up to multiples of $2 \pi {\rm i}$. To be equivalent to the quad equations, the additive three-leg form needs to include an integer multiple of~$2 \pi {\rm i}$. We will make such a term appear in the corner equations by including an integer-valued field in the action.

An important property of Lagrangian multiform theory is that, on solutions to the variational principle, the Lagrangian $d$-form is closed. This means in particular that the action over a~closed surface vanishes on solutions. The existing proofs of closure make implicit assumptions regarding the branch choices in the logarithm or dilogarithm functions in the action. We give a counterexample showing that, with the standard choice of branches, the closure property fails on some solutions. We remedy this in a similar way to the three-leg form issue: we add integer fields to the action that balance out the branch choices in the action. For equations H1, H2, A1, Q1, and Q2 we recover the result that the action around an elementary cube is zero. However, for H3, A2, and Q3, we only show that it is a multiple of $4 \pi^2$. We present numerical evidence that non-zero multiples of $4 \pi^2$ do occur. While we do not discuss equation Q4 here, we believe that the techniques we developed provide a possible path to establishing its closure relation.\looseness=1

The structure of this paper is as follows. In Section~\ref{sec-abs}, we review the structure of the ABS equations, emphasising the multiples of $2 \pi {\rm i}$ that must be included in the three-leg forms to ensure they are equivalent to the quad equations. We start Section~\ref{sec-multiforms} with a short introduction to Lagrangian multiform theory. Then, in Section~\ref{sec-trident} we introduce the 4-point (\emph{trident}) Lagrangian that provides a variational structure for the quad equations. In Section~\ref{sec-triangle}, we~contrast this to the more common 3-point (\emph{triangle}) Lagrangian, which produces weaker corner equations. In Section~\ref{subsec-clos}, we~revisit the closure property. We give a proof of closure based on deforming a given solution into a very simple one for which closure holds, while tracking the effects of potentially crossing branch cuts during the deformation. These arguments take place on an elementary cube; in Section~\ref{sec-arbitrary}, we comment on how they generalise to arbitrary discrete surfaces. Finally, Section~\ref{sec-conclusion} provides some conclusions and perspectives.\looseness=1

We performed some symbolic and numerical verifications of the computations presented in this paper in SageMath \cite{sagemath}. The code can be found at \cite{code}.

\newpage

\section{Quad equations}
\label{sec-abs}

\subsection{The ABS classification}

We consider partial difference equations in the lattice $\Z^N$, where to each lattice direction there is an associated parameter $\alpha_i \in \C$. For a function $u\colon \Z^N \to \C$ and a reference point $\vec n \in \Z^N$, we use the shorthand notations $u = u(\vec n)$, $u_i = u(\vec n + \vec e_i)$, $u_{ij} = u(\vec n + \vec e_i + \vec e_j)$, etc., where $\vec e_i$ is the unit vector in the $i$-th coordinate direction.

The ABS list \cite{adler2003classification} is a classification of integrable difference equations that satisfy the following conditions:
\begin{itemize}\itemsep=0pt
 \item They are quad equations: they depend on a square stencil and on two parameters associated to the two lattice directions, i.e., they take the form
 \begin{equation}
 \label{quad}
 Q(u,u_i,u_j,u_{ij},\alpha_i,\alpha_j) = 0 .
 \end{equation}

 \item They are symmetric: equation \eqref{quad} is invariant under the symmetries of the square,
 \begin{equation}
 \label{quad-symmetries}
 Q(u,u_i,u_j,u_{ij},\alpha_i,\alpha_j)
 = \pm Q(u,u_j,u_i,u_{ij},\alpha_j,\alpha_i)
 = \pm Q(u_i,u,u_{ij},u_{j},\alpha_i,\alpha_j) .
 \end{equation}
 Note that these two transformations of the tuple $(u,u_i,u_j,u_{ij})$ generate the symmetry group $D_4$, and that the $\alpha_i$, $\alpha_j$ are interchanged accordingly.

 \item They are multi-affine: the function $Q$ is a polynomial of degree one in each of the $u$,~$u_i$, $u_j$,~$u_{ij}$ (but can be of higher total degree). This guarantees that we can solve the equations for any of the four variables, given the other three.

 \item They are three-dimensionally consistent: given initial values $u$, $u_i$, $u_j$, $u_k$, use the equations%
 \begin{align*}
 &Q_{ij}:= Q(u,u_i,u_j,u_{ij},\alpha_i,\alpha_j) = 0 , \\
 &Q_{jk}:= Q(u,u_j,u_k,u_{jk},\alpha_j,\alpha_k) = 0 , \\
 &Q_{ki}:= Q(u,u_k,u_i,u_{ki},\alpha_k,\alpha_i) = 0 ,
 \end{align*}
 to determine $u_{ij}$, $u_{jk}$, $u_{ki}$; then $u_{ijk}$ should be uniquely determined by
 \begin{align*}
 &\inversion Q_{ij}:= Q(u_{ijk},u_{jk},u_{ki},u_{k},\alpha_i,\alpha_j) = 0 , \\
 &\inversion Q_{jk}:= Q(u_{ijk},u_{ki},u_{ij},u_{i},\alpha_j,\alpha_k) = 0 , \\
 &\inversion Q_{ki}:= Q(u_{ijk},u_{ij},u_{jk},u_{j},\alpha_k,\alpha_i) = 0 .
 \end{align*}
 This is illustrated in Figure \ref{fig-cac}.

 \item They satisfy the tetrahedron property: the value obtained for $u_{ijk}$ in the computation above is independent of $u$, depending only on $u_i$, $u_j$, $u_k$ (and on the lattice parameters).\looseness=1
\end{itemize}

\begin{figure}[t]
\centering
\begin{tikzpicture}[scale=2]
	\begin{scope}[shift=({0,0}), y={(5mm, 3mm)}, z={(0cm,1cm)}]
		\draw
		(0,0,0) node[below left] {$u$} -- (1,0,0) node[below] {$u_i$} -- (1,1,0) node[right] {$u_{ij}$};
 \draw[dashed] (1,1,0) -- (0,1,0) node[above left] {$u_j$} -- (0,0,0);
 \draw[dashed] (0,1,0) -- (0,1,1) node[above] {$u_{jk}$};
 \draw (0,1,1) -- (0,0,1) node[left] {$u_k$} -- (0,0,0)
 (1,1,1) node[above] {$u_{ijk}$} -- (0,1,1);
 \draw (1,1,0) -- (1,1,1);
 \node at (0,0,0) {$\bullet$};
 \node at (1,0,0) {$\bullet$};
 \node at (0,1,0) {$\bullet$};
 \node at (0,0,1) {$\bullet$};
 \node[gray] at (1,1,0) {$\bullet$};
 \node[gray] at (0,1,1) {$\bullet$};
 \node[gray!50] at (1,1,1) {$\bullet$};
	\end{scope}
 \begin{scope}[shift=({2.5,0}), y={(5mm, 3mm)}, z={(0cm,1cm)}]
		\draw (0,0,0) node[below left] {$u$} -- (1,0,0) node[below] {$u_i$} -- (1,1,0) node[right] {$u_{ij}$};
 \draw[dashed] (1,1,0) -- (0,1,0) node[above] {$u_j$} -- (0,0,0);
 \draw (0,0,0) -- (0,0,1) node[left] {$u_k$} -- (1,0,1) node[below right] {$u_{ki}$} -- (1,0,0) -- (0,0,0)
 (1,1,1) node[above] {$u_{ijk}$} -- (1,0,1) -- (1,0,0) -- (1,1,0) -- (1,1,1) ;
 \node at (0,0,0) {$\bullet$};
 \node at (1,0,0) {$\bullet$};
 \node at (0,1,0) {$\bullet$};
 \node at (0,0,1) {$\bullet$};
 \node[gray] at (1,1,0) {$\bullet$};
 \node[gray] at (1,0,1) {$\bullet$};
 \node[gray!50] at (1,1,1) {$\bullet$};
	\end{scope}	
 \begin{scope}[shift=({5,0}), y={(5mm, 3mm)}, z={(0cm,1cm)}]
		\draw[dashed] (0,0,0) node[below left] {$u$} -- (0,1,0) node[above left] {$u_j$} -- (0,1,1) node[above] {$u_{jk}$};
 \draw (0,1,1) -- (0,0,1) node[left] {$u_k$} -- (0,0,0)
 (0,0,0) -- (0,0,1) -- (1,0,1) node[below right] {$u_{ki}$} -- (1,0,0) node[below] {$u_i$} -- (0,0,0)
 (1,1,1) node[above] {$u_{ijk}$} -- (1,0,1) -- (0,0,1) -- (0,1,1) -- (1,1,1) ;
 \node at (0,0,0) {$\bullet$};
 \node at (1,0,0) {$\bullet$};
 \node at (0,1,0) {$\bullet$};
 \node at (0,0,1) {$\bullet$};
 \node[gray] at (1,0,1) {$\bullet$};
 \node[gray] at (0,1,1) {$\bullet$};
 \node[gray!50] at (1,1,1) {$\bullet$};
	\end{scope}	
\end{tikzpicture}
\caption{Multidimensional consistency demands that all three of these routes to calculate $u_{ijk}$ from initial values $u$, $u_i$, $u_j$, $u_k$ produce the same value.}
\label{fig-cac}
\end{figure}
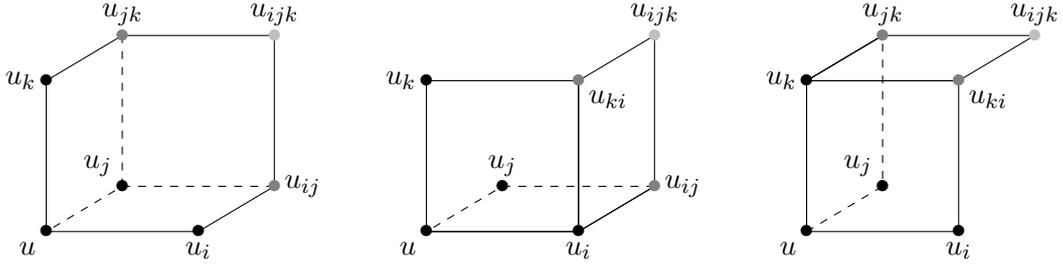

\begin{Remark}
 Note that, regardless of which signs occur in equation \eqref{quad-symmetries}, we have that
 \[ \inversion Q_{ij} := Q(u_{ijk},u_{jk},u_{ki},u_k,\alpha_i,\alpha_j) =Q(u_k,u_{ki},u_{jk},u_{ijk},\alpha_i,\alpha_j), \]
 because the last equality is obtained by applying both symmetries of the square twice. We will use the notation \smash{\raisebox{-1pt}{$\reverse$}} more generally to denote the \emph{point inversion} $u \leftrightarrow u_{ijk}$, $u_i \leftrightarrow u_{jk}$, etc.
\end{Remark}

The ABS list is given in Appendix~\ref{AppendixA}. By convention, its members are denoted H1--H3, A1--A2, and Q1--Q4. Some of them depend on a parameter $\delta$, and we use a subscript on the equation name to indicate whether or not this parameter is zero.

Each equation $Q(u,u_i,u_j,u_{ij}, \alpha_i, \alpha_j) = 0$ of the ABS list can be written in a \emph{three-leg form}. For H1, A1$_{\delta=0}$, and Q1$_{\delta=0}$, there exist functions $\psi$ and $\phi$ such that the multi-affine equation is equivalent to
\begin{equation}
\label{three-leg-0}
\phi(u,u_{ij}, \alpha_i-\alpha_j) = \psi(u,u_i,\alpha_i) - \psi(u,u_j,\alpha_j) .
\end{equation}
For the other quad equations, the most natural three-leg form has a multiplicative rather than an additive structure: there exist functions $\Psi$ and $\Phi$ such that the multi-affine equation is equivalent to
\[ \Phi(u,u_{ij}, \alpha_i-\alpha_j) = \frac{ \Psi(u,u_i,\alpha_i) }{ \Psi(u,u_j,\alpha_j) } . \]
We can bring this into an additive form by taking the logarithm. Since we are working over~$\C$, writing the logarithm of a quotient as the difference of logarithms may introduce an error of~$2 \pi {\rm i}$. 
So, putting $\psi = \log(\Psi)$ and $\phi = \log(\Phi)$, we find that (for H2, H3, A1$_{\delta\neq0}$, A2, Q1$_{\delta\neq0}$, Q2, Q3, Q4) the multi-affine equation is equivalent to
\begin{equation}
\label{three-leg-theta}
\phi(u,u_{ij}, \alpha_i-\alpha_j) = \psi(u,u_i,\alpha_i) - \psi(u,u_j,\alpha_j) + 2 \Theta \pi {\rm i} , \qquad \Theta \in \Z .
\end{equation}
Based on the discussion above, we could specify $\Theta \in \{-1,0,1\}$ in equation \eqref{three-leg-theta}. However, in~many cases, the $\psi$ and $\phi$ given in Appendix~\ref{AppendixA} are simplified by applying further logarithm identities to the right-hand sides of $\psi = \log(\Psi)$ and $\phi = \log(\Phi)$, which can lead to further multiples of $2 \pi {\rm i}$ being added. To accommodate this, we allow $\Theta \in \Z$.

\begin{Remark}
As an alternative to introducing the integer $\Theta$, we could use non-principal branches of the logarithm when putting $\psi = \log(\Psi)$ and $\phi = \log(\Phi)$. Then for every solution $(u,u_i,u_j,u_{ij})$ to the multi-affine quad equation, there exist a choice of branch for the logarithm $\psi = \log(\Psi)$ on the diagonal leg, such that
\[ \phi(u,u_{ij}, \alpha_i-\alpha_j) = \psi(u,u_i,\alpha_i) - \psi(u,u_j,\alpha_j) . \]
(In the expressions $\phi = \log(\phi)$ for the short legs, we can still fix the principal branch.)

Most papers on the subject mention neither the choice of branch cuts, nor the possibility of a term $2 \Theta \pi {\rm i}$ in the three-leg form. A notable exception, where these subtleties are explicitly addressed, is \cite{bobenko2012discrete}.
\end{Remark}

We denote the additive three-leg expression, based at the vertex $u$, by
\begin{equation}
	\mathcal{Q}_{ij}^{(u)} := \psi(u,u_i,\alpha_i) - \psi(u,u_j,\alpha_j) - \phi(u,u_{ij}, \alpha_i-\alpha_j) .
	\label{quad-u}
\end{equation}
We can also consider three-leg forms based at the other three vertices of the square:
\begin{align}
	&\mathcal{Q}_{ij}^{(u_i)}:= \psi(u_i,u_{ij},\alpha_j) - \psi(u_i,u,\alpha_i) - \phi(u_i,u_j, \alpha_j-\alpha_i) , \label{quad-ui}\\
	&\mathcal{Q}_{ij}^{(u_{ij})}:= \psi(u_{ij},u_j,\alpha_i) - \psi(u_{ij},u_i,\alpha_j) - \phi(u_{ij},u, \alpha_i-\alpha_j) , \label{quad-uij}\\
	&\mathcal{Q}_{ij}^{(u_j)}:= \psi(u_j,u,\alpha_j) - \psi(u_j,u_{ij},\alpha_i) - \phi(u_j,u_i, \alpha_j-\alpha_i) . \label{quad-uj}
\end{align}
The functions $\phi$ and $\psi$ are independent of the choice of base vertex. This reflects the rotational symmetry of the multi-affine quad equation. The three-leg forms based at different vertices are illustrated in Figure \ref{fig:3leg}.

Similarly, we denote the three-leg forms of \smash{$\inversion{Q}_{ij}$} by
\begin{align}
&	\mathcal{Q}_{ij}^{(u_{ijk})}:= \inversion{\mathcal{Q}_{ij}^{(u)}} = \psi(u_{ijk},u_{jk},\alpha_i) - \psi(u_{ijk},u_{ki},\alpha_j) - \phi(u_{ijk},u_{k}, \alpha_i-\alpha_j) , \label{quad-uijk} \\
&	\mathcal{Q}_{ij}^{(u_{jk})}:= \inversion{\mathcal{Q}_{ij}^{(u_i)}} = \psi(u_{jk},u_{k},\alpha_j) - \psi(u_{jk},u_{ijk},\alpha_i) - \phi(u_{jk},u_{ki}, \alpha_j-\alpha_i) , \label{quad-ujk}\\
&	\mathcal{Q}_{ij}^{(u_{k})}:= \inversion{\mathcal{Q}_{ij}^{(u_{ij})}} = \psi(u_{k},u_{ki},\alpha_i) - \psi(u_{k},u_{jk},\alpha_j) - \phi(u_{k},u_{ijk}, \alpha_i-\alpha_j) , \label{quad-uk}\\
&	\mathcal{Q}_{ij}^{(u_{ki})}:= \inversion{\mathcal{Q}_{ij}^{(u_j)}} = \psi(u_{ki},u_{ijk},\alpha_j) - \psi(u_{ki},u_{k},\alpha_i) - \phi(u_{ki},u_{jk}, \alpha_j-\alpha_i) . \label{quad-uki}
\end{align}
The convention we use in these notations is that the upper index denotes the vertex of the cube that the three-leg form is based at, and the lower indices indicate the directions spanning the face of the cube it is situated in.

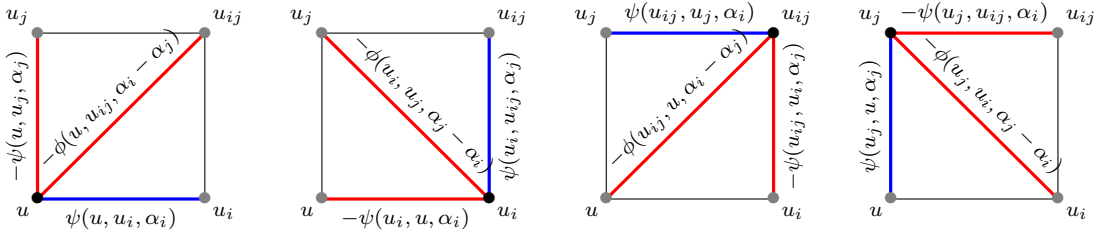
\begin{figure}[t]
\scriptsize
\centering
\begin{tikzpicture}[scale=2.2]
	\draw (0,0) node[below left] {$u$} -- (1,0) node[below right] {$u_i$} -- (1,1) node[above right] {$u_{ij}$} -- (0,1) node[above left] {$u_j$} -- cycle;
	\begin{scope}[very thick]
		\draw[blue] (0,0) -- node[below] {\color{black}$\psi(u,u_i,\alpha_i)$} (1,0);
		\draw[red] (0,1) -- node[above, rotate=90] {\color{black}$-\psi(u,u_j,\alpha_j)$} (0,0);
		\draw[red] (1,1) -- node[above, rotate=45] {\color{black}\ $-\phi(u,u_{ij}, \alpha_i-\alpha_j)$} (0,0);
		\node at (0,0) {\normalsize $\bullet$};
		\node[gray] at (1,0) {\normalsize $\bullet$};
		\node[gray] at (0,1) {\normalsize $\bullet$};
		\node[gray] at (1,1) {\normalsize $\bullet$};
	\end{scope}
	\begin{scope}[shift=({1.7,0})]
		\draw (0,0) node[below left] {$u$} -- (1,0) node[below right] {$u_i$} -- (1,1) node[above right] {$u_{ij}$} -- (0,1) node[above left] {$u_j$} -- cycle;
		\begin{scope}[very thick]
			\draw[blue] (1,0) -- node[below, rotate=90] {\color{black}$\psi(u_i,u_{ij},\alpha_j)$} (1,1);
			\draw[red] (0,0) -- node[below] {\color{black}$-\psi(u_i,u,\alpha_i)$} (1,0);
			\draw[red] (0,1) -- node[above, rotate=-45] {\color{black}\ $-\phi(u_i,u_j, \alpha_j-\alpha_i)$} (1,0);
			\node at (1,0) {\normalsize $\bullet$};
 		\node[gray] at (0,0) {\normalsize $\bullet$};
 		\node[gray] at (0,1) {\normalsize $\bullet$};
 		\node[gray] at (1,1) {\normalsize $\bullet$};
		\end{scope}
	\end{scope}
	\begin{scope}[shift=({3.4,0})]
		\draw (0,0) node[below left] {$u$} -- (1,0) node[below right] {$u_i$} -- (1,1) node[above right] {$u_{ij}$} -- (0,1) node[above left] {$u_j$} -- cycle;
		\begin{scope}[very thick]
			\draw[blue] (1,1) -- node[above] {\color{black}$\psi(u_{ij},u_j,\alpha_i)$} (0,1);
			\draw[red] (1,0) -- node[below, rotate=90] {\color{black}$-\psi(u_{ij},u_i,\alpha_j)$} (1,1);
			\draw[red] (0,0) -- node[above, rotate=45] {\color{black}\ $-\phi(u_{ij},u, \alpha_i-\alpha_j)$} (1,1);
			\node at (1,1) {\normalsize $\bullet$};
 		 \node[gray] at (1,0) {\normalsize $\bullet$};
 		\node[gray] at (0,1) {\normalsize $\bullet$};
 		\node[gray] at (0,0) {\normalsize $\bullet$};
		\end{scope}
	\end{scope}
	\begin{scope}[shift=({5.1,0})]
		\draw (0,0) node[below left] {$u$} -- (1,0) node[below right] {$u_i$} -- (1,1) node[above right] {$u_{ij}$} -- (0,1) node[above left] {$u_j$} -- cycle;
		\begin{scope}[very thick]
			\draw[blue] (0,1) -- node[above, rotate=90] {\color{black}$\psi(u_j,u,\alpha_j)$} (0,0);
			\draw[red] (1,1) -- node[above] {\color{black}$-\psi(u_j,u_{ij},\alpha_i)$} (0,1);
			\draw[red] (1,0) -- node[above, rotate=-45] {\color{black}\ $-\phi(u_j,u_i, \alpha_j-\alpha_i)$} (0,1);
			\node at (0,1) {\normalsize $\bullet$};
 \node[gray] at (1,0) {\normalsize $\bullet$};
 		\node[gray] at (0,0) {\normalsize $\bullet$};
 		\node[gray] at (1,1) {\normalsize $\bullet$};
		\end{scope}
	\end{scope}
\end{tikzpicture}
\caption{Graphical representation of the four orientations of a three-leg form, where the colour reflects the sign of the leg function.}
\label{fig:3leg}
\end{figure}

\begin{Proposition}
For ${\rm H}2$, ${\rm H}3$, ${\rm A}1_{\delta=1}$, ${\rm A}2$, ${\rm Q}1_{\delta=1}$, ${\rm Q}2$, ${\rm Q}3$, ${\rm Q}4$, the following are equivalent:
\begin{enumerate}\itemsep=0pt
 \item[$1.$] The multi-affine quad equations holds, i.e., $Q_{ij} = 0$.
 \item[$2.$] There exists a $\Theta \in \Z$ such that \smash{$\mathcal{Q}_{ij}^{(u)} = 2 \Theta \pi {\rm i}$}.
 \item[$3.$] There exists a $\Theta \in \Z$ such that \smash{$\mathcal{Q}_{ij}^{(u_i)} = 2 \Theta \pi {\rm i}$}.
 \item[$4.$] There exists a $\Theta \in \Z$ such that \smash{$\mathcal{Q}_{ij}^{(u_j)} = 2 \Theta \pi {\rm i}$}.
 \item[$5.$] There exists a $\Theta \in \Z$ such that \smash{$\mathcal{Q}_{ij}^{(u_{ij})} = 2 \Theta \pi {\rm i}$}.
\end{enumerate}
For ${\rm H}1$, ${\rm A}1_{\delta=0}$, ${\rm Q}1_{\delta=0}$, we have
\[ Q_{ij} = 0 \ \Leftrightarrow\ \mathcal{Q}_{ij}^{(u)} = 0 \ \Leftrightarrow\ \mathcal{Q}_{ij}^{(u_i)} = 0 \ \Leftrightarrow\ \mathcal{Q}_{ij}^{(u_j)} = 0 \ \Leftrightarrow\ \mathcal{Q}_{ij}^{(u_{ij})} = 0 . \]
\end{Proposition}

\begin{ExampleH2}
Consider equation H2, for which
\[ Q_{ij} = (u - u_{ij}) (u_i - u_j) - (\alpha_i - \alpha_j) (u + u_i + u_{ij} + u_j) -\alpha_i^2 + \alpha_j^2 . \]
Its multiplicative three-leg form is given by
\[ \frac{\alpha_i - \alpha_j + u - u_{ij}}{-\alpha_i + \alpha_j + u - u_{ij}} \cdot \frac{\alpha_i + u + u_i}{\alpha_j + u + u_j} = 1 . \]
Hence, the additive three-leg form is given by
\[ \log(\alpha_i + u + u_i) - \log(\alpha_j + u + u_j) - \log \left( \frac{-\alpha_i + \alpha_j + u - u_{ij}}{\alpha_i - \alpha_j + u - u_{ij}} \right) = 2 \Theta \pi {\rm i} . \]
(Note that the $\phi$ given in Appendix~\ref{AppendixA} expands the logarithm of the quotient as a difference of logarithms, which could be off by $2 \pi {\rm i}$ and hence could change the value of $\Theta$.)

To illustrate the importance of the right-hand side of the three leg form, consider $u$, $u_i$, $u_j$, $u_{ij}$, $\alpha_i$, $\alpha_j$ that satisfy
\begin{align*}
	&\alpha_i + u + u_i = \exp\bigl(\tfrac23 \pi {\rm i}\bigr) , \\
	&\alpha_j + u + u_j = \exp\bigl(-\tfrac23 \pi {\rm i}\bigr) , \\
	&\frac{-\alpha_i + \alpha_j + u - u_{ij}}{\alpha_i - \alpha_j + u - u_{ij}} = \exp\bigl(-\tfrac23 \pi {\rm i}\bigr) .
\end{align*}
This is a solution to the $Q_{ij} = 0$ because it manifestly solves the equation in multiplicative three-leg form. However, the additive three-leg form does not equal zero:
\[
	\log(\alpha_i + u + u_i) - \log(\alpha_j + u + u_j) - \log \left( \frac{-\alpha_i + \alpha_j + u - u_{ij}}{\alpha_i - \alpha_j + u - u_{ij}} \right)
	= 2 \pi {\rm i} .
\]
\end{ExampleH2}

On the faces of a cube adjacent to a fixed vertex, the three-leg forms based at that vertex naturally combine to form an equation on a tetrahedral stencil, illustrated in Figure \ref{fig-tetra}, for example,
\begin{align}
 \mathcal{T}^{(u)}:={}& \phi(u,u_{ij},\alpha_i-\alpha_j) + \phi(u,u_{jk},\alpha_j-\alpha_k) + \phi(u,u_{ki},\alpha_k-\alpha_i) \label{tetra-u}\\
 ={}& {-}\mathcal{Q}^{(u)}_{ij} - \mathcal{Q}^{(u)}_{jk} - \mathcal{Q}^{(u)}_{ki} . \notag
\end{align}
Then the multi-affine quad equations imply that $\mathcal{T}^{(u)}$ is a multiple of $2 \pi {\rm i}$.
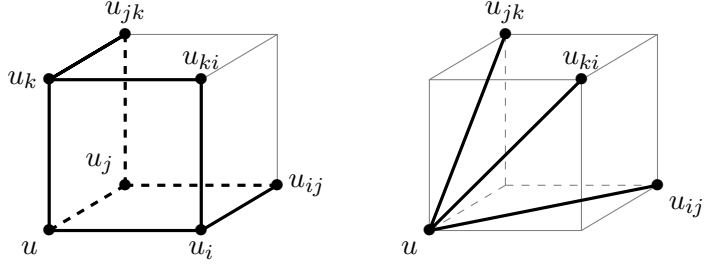
\begin{figure}[t]
\centering
\begin{tikzpicture}[scale=2]
	\begin{scope}[y={(5mm, 3mm)}, z={(0cm,1cm)}]
 \draw[gray, every edge/.append style={dashed}]
		(1,0,0) -- (0,0,0) -- (0,0,1) -- (0,1,1) -- (1,1,1) -- (1,1,0) -- cycle -- (1,0,1) -- (0,0,1)
		(1,0,1) -- (1,1,1)
		(0,1,0) edge (0,0,0) edge (0,1,1) edge (1,1,0);
		\draw[very thick, every edge/.append style={dashed}]
		(1,1,0) node[right] {$u_{ij}$} -- (1,0,0) node[below] {$u_i$} -- (0,0,0) node[below left] {$u$} -- (0,0,1) node[left] {$u_k$} -- (0,1,1) node[above] {$u_{jk}$} -- (0,0,1) -- (1,0,1) node[above] {$u_{ki}$} -- (1,0,0)
		(0,1,0) edge (0,0,0) edge (0,1,1) edge (1,1,0) node[above left] {$u_j$};
 \node at (0,0,0) {$\bullet$};
 \node at (1,0,0) {$\bullet$};
 \node at (0,1,0) {$\bullet$};
 \node at (0,0,1) {$\bullet$};
 \node at (1,1,0) {$\bullet$};
 \node at (0,1,1) {$\bullet$};
 \node at (1,0,1) {$\bullet$};
	\end{scope}
 \begin{scope}[shift=({2.5,0}),scale=1, y={(5mm, 3mm)}, z={(0cm,1cm)}]
 	\draw[gray, every edge/.append style={dashed}]
		(1,0,0) -- (0,0,0) -- (0,0,1) -- (0,1,1) -- (1,1,1) -- (1,1,0) -- cycle -- (1,0,1) -- (0,0,1)
		(1,0,1) -- (1,1,1)
		(0,1,0) edge (0,0,0) edge (0,1,1) edge (1,1,0);
		\draw[very thick, every edge/.append style={dashed}]
 (0,0,0) node[below left] {$u$} -- (0,1,1) node[above] {$u_{jk}$}
		(0,0,0) -- (1,0,1) node[above] {$u_{ki}$}
		(0,0,0) -- (1,1,0) node[below right] {$u_{ij}$};
 \node at (0,0,0) {$\bullet$};
 \node at (1,1,0) {$\bullet$};
 \node at (0,1,1) {$\bullet$};
 \node at (1,0,1) {$\bullet$};
	\end{scope}	
\end{tikzpicture}
\caption{Tetrahedron equation from three quad equations.}
\label{fig-tetra}
\end{figure}

To relate the three-leg tetrahedron equation \eqref{tetra-u} to a multi-affine equation, we rely on some relations between the equations of the ABS list. The ABS equations of type Q have the property that their short and long leg functions are the same: $\phi=\psi$. Each equation of type~H or~A shares its long leg function $\phi$ with an equation of type~Q, but has a different short leg function $\psi$.
Thus, the quad equations of type~H and type~A have the same the tetrahedron equation as an equation of type Q. In fact, a multi-affine tetrahedron equation $T=0$, equivalent to $\mathcal{T}^{(u)} = 0$ or $\mathcal{T}^{(u)} = 2 \Theta \pi {\rm i}$, can be stated in terms of a quad polynomial of type Q, evaluated on a tetrahedron stencil:
\[ T(u,u_{ij},u_{jk},u_{ki},\alpha_i,\alpha_j,\alpha_k)
 = Q^{(\text{type Q})}(u,u_{ij},u_{jk},u_{ki},\alpha_i-\alpha_j,-\alpha_j+\alpha_k) . \]
Here $Q^{(\text{type Q})}$ represents the multi-affine polynomial corresponding to a quad equation of type~Q.

The same argument applies to the other four variables, so we can derive a second tetrahedron equation in three-leg form, which is related to the first by point inversion:
\begin{equation}
 \label{tetra-uijk}
 \mathcal{T}^{(u_{ijk})} := \phi(u_{ijk},u_{k},\alpha_i-\alpha_j) + \phi(u_{ijk},u_{i},\alpha_j-\alpha_k) + \phi(u_{ijk},u_{j},\alpha_k-\alpha_i)
\end{equation}

\begin{ExampleH1}
One of the simplest equations on the ABS list is the lattice potential KdV equation, labelled H1, for which
\[
 Q_{ij} = (u_i - u_j)(u - u_{ij}) - \alpha_i + \alpha_j.
\]
In this case, the three-leg form is easily found in additive form, with
\begin{equation}
\label{H1-three-leg}
\begin{split}
 \psi(u,u_i,\alpha_i) = u_i
 \qquad \text{and} \qquad
 \phi(u,u_{ij},\alpha_i-\alpha_j) = \frac{\alpha_i - \alpha_j}{u - u_{ij}} .
\end{split}
\end{equation}
The quad equation $Q_{ij} = 0$ is equivalent to the three-leg equation \smash{$\mathcal{Q}^{(u)}_{ij} = 0$}, where
\begin{align*}
 \mathcal{Q}^{(u)}_{ij} & = \psi(u,u_i,\alpha_i) - \psi(u,u_j,\alpha_j) - \phi(u,u_{ij},\alpha_i-\alpha_j) \notag \\
 & = u_i - u_j - \frac{\alpha_i - \alpha_j}{u - u_{ij}} = \frac{Q_{ij}}{u - u_{ij}} .
\end{align*}
Note that there is no $2 \Theta \pi {\rm i}$ term, because for H1 we do not need to take a logarithm to arrive at an additive three-leg form.

The three-leg form $\mathcal{T}^{(u)} = 0$ of the tetrahedron equation is given by
\begin{align*}
 \mathcal{T}^{(u)}
 & = \phi(u,u_{ij},\alpha_i-\alpha_j) + \phi(u,u_{jk},\alpha_j-\alpha_k) + \phi(u,u_{ki},\alpha_k-\alpha_i) \\
 &=\frac{\alpha_i - \alpha_j}{u - u_{ij}} + \frac{\alpha_j - \alpha_k}{u - u_{jk}} + \frac{\alpha_k - \alpha_i}{u - u_{ki}} .
\end{align*}
The tetrahedron equation can also be expressed in multi-affine form as $T=0$ with
\begin{gather*}
 T(u,u_{ij},u_{jk},u_{ki},\alpha_i,\alpha_j,\alpha_k) \\
 \qquad{} = (\alpha_i - \alpha_j)(u u_{ij} + u_{jk}u_{ki}) + (\alpha_j - \alpha_k)(u u_{jk} + u_{ij} u_{ki}) + (\alpha_k - \alpha_i)(u u_{ki} + u_{ij} u_{jk} ) ,
\end{gather*}
which is exactly the multi-affine polynomial of ABS equation Q1$_{\delta=0}$ (see Appendix~\ref{AppendixA}). The relation between the two forms of the tetrahedron equation is given by
\[ \mathcal{T}^{(u)} = \frac{T}{(u - u_{ij}) (u - u_{jk}) (u - u_{jk})} . \]
\end{ExampleH1}

\subsection{Planar Lagrangian structure}
\label{subsec-planar}

Quad equations are an example of 2-dimensional partial difference equations. The traditional point of view on variational principles for such equations is that they minimise a sum, over the planar lattice $\Z^2$, of some discrete Lagrangian. Unlike Lagrangian multiform theory, this does not capture the multidimensional consistency of these equations, but finding planar Lagrangian structures is a useful initial step towards studying Lagrangian multiforms.

Already in the original ABS paper \cite{adler2003classification}, Lagrangians in this sense were constructed for all equations from the ABS list. This construction is based on the observation that, after a suitable variable transformation, there exist functions $L$ and $\Lambda$ such that the leg functions $\psi$ and $\phi$ can be alternatively expressed as
\begin{align}
	&\psi(x,y, \alpha)= \pdv{}{x} L(x,y,\alpha) , \nonumber\\
	&\phi(x,y, \alpha-\beta)= \pdv{}{x} \Lambda(x,y,\alpha-\beta) ,\label{antiderivative1}
\end{align}
and as
\begin{align}
	&\psi(y,x,\alpha)= \pdv{}{y} L(x,y,\alpha) ,\nonumber \\
	&\phi(y,x, \alpha- \beta)= \pdv{}{y} \Lambda(x,y,\alpha- \beta) .\label{antiderivative2}
\end{align}
These conditions suggest that $L$ and $\Lambda$ may have the symmetries
\begin{align}
	&L(x,y,\alpha)= L(y,x, \alpha) , \nonumber\\
	&\Lambda(x,y,\alpha-\beta)= \Lambda(y,x,\alpha-\beta) ,\nonumber \\
	&\Lambda(x,y,\alpha-\beta)= -\Lambda(x,y,\beta-\alpha) ,\label{L-symmetries}
\end{align}
but this is not a necessary condition. While it is possible to find $L$ and $\Lambda$ satisfying \eqref{L-symmetries}, it is often more convenient to work with non-symmetric versions. In Appendix~\ref{AppendixA}, we list both options where appropriate.
The existence of these functions $L$ and $\Lambda$ was shown by implicit arguments in \cite{adler2003classification}. For most of the ABS equations, explicit expressions for $L$ and $\Lambda$ were found in \cite{lobb2009lagrangian, vermeeren2019variational}.

\begin{ExampleH1}
 We can integrate the leg functions \eqref{H1-three-leg} for H1 to find
 \begin{align}
 & L(u,u_i,\alpha_i)= u u_i ,\nonumber \\
 & \Lambda(u,u_{ij},\alpha_i-\alpha_j)= (\alpha_i - \alpha_j) \log(u - u_{ij}) . \label{H1-L-Lambda}
 \end{align}
 Note that $\Lambda$ is not quite symmetric under exchange of $u$ and $u_{ij}$, but these functions do satisfy the conditions \eqref{antiderivative1} and \eqref{antiderivative2}.
\end{ExampleH1}

\begin{ExampleH3}
 $Q_{ij}(u,u_i,u_{ij},u_j,\alpha_i,\alpha_j) = (u u_i + u_{ij} u_j) \alpha_i - (u_i u_{ij} + u u_j) \alpha_j$ has a~multiplicative three-leg form with
 \begin{align*}
 \Psi(u,u_i,\alpha_i) = u_i , \qquad
 \Phi(u,u_{ij},\alpha_i-\alpha_j) = \frac{\alpha_i \bigl(1 - \frac{\alpha_j u}{\alpha_i u_{ij}}\bigr) }{ \alpha_j \bigl( 1 - \frac{\alpha_i u}{\alpha_j u_{ij}} \bigr)} ,
 \end{align*}
 and hence an additive three-leg form with
 \begin{align*}
 &\psi(u,u_i,\alpha_i)= \log(u_i) , \\
 &\phi(u,u_{ij},\alpha_i-\alpha_j)= \log \biggl(1 - \frac{\alpha_j u}{\alpha_i u_{ij}}\biggr) - \log \biggl( 1 - \frac{\alpha_i u}{\alpha_j u_{ij}} \biggr) + \log(\alpha_i) - \log(\alpha_j) .
 \end{align*}
 Note that the properties \eqref{antiderivative1} and \eqref{antiderivative2} imply that the partial derivative of $\psi$ with respect to its second entry must be a symmetric function of $u$ and $u_i$:
 \[ \pdv{\psi(u,u_i,\alpha_i)}{u_i} = \pdv{^2 L(u,u_i,\alpha_i)}{u \partial u_i} = \pdv{\psi(u_i,u,\alpha_i)}{u} , \]
 which is not the case for $\psi(u,u_i,\alpha_i) = \log(u_i)$. Hence, this three-leg form is not suitable to construct a Lagrangian. We can bring it into a suitable form with the transformation $u = {\rm e}^U$. It will be convenient to also transform the parameters and set $\alpha = {\rm e}^A$. Then we have a three-leg form with
 \begin{align*}
 &\psi(U,U_i,A_i)= -U_i , \\
 &\phi(U,U_{ij},A_i-A_j)= \log \bigl( 1 - {\rm e}^{A_i-A_j+U-U_{ij}} \bigr) - \log \bigl(1 - {\rm e}^{-A_i+A_j+U-U_{ij}} \bigr)- A_i + A_j ,
 \end{align*}
 where we have introduced an overall minus sign that is irrelevant here, but makes sure that these expressions will be consistent with Lemma~\ref{lemma-dSdalpha} of Section~\ref{subsec-clos}.
 For these functions, we can find $L$ and $\Lambda$ satisfying \eqref{antiderivative1} and \eqref{antiderivative2}:
 \begin{align}
 &L(U,U_i,A_i) = -U U_i \label{H3-L-Lambda}\\
 &\Lambda(U,U_{ij},A_i-A_j) = \dilog \bigl( {\rm e}^{-A_i+A_j+U-U_{ij}} \bigr) - \dilog \bigl( {\rm e}^{A_i-A_j+U-U_{ij}} \bigr) - (A_i-A_j) (U-U_{ij}) ,\nonumber
 \end{align}
 where $\dilog$ denotes the dilogarithm function, which satisfies
 \[
 \pdv{\dilog(z)}{z} = - \frac{\log(1-z)}{z}.
 \]
\end{ExampleH3}

The functions $L$ and $\Lambda$ can be combined into a 4-point Lagrangian \cite{adler2003classification}
\begin{equation}
 \label{L_trident}
 \cL_\trileg(U,U_i,U_j,U_{ij},A_i,A_j) := L(U,U_i,A_i) - L(U,U_j,A_j) - \Lambda(U,U_{ij}, A_i-A_j) ,
\end{equation}
where we have written the fields and parameters with capital letters to reflect the fact that a~variable transformation may be necessary and the quad equation may not be multi-affine in these variables. We refer to the Lagrangian \eqref{L_trident} as the \emph{trident Lagrangian}, inspired by the three-leg structure depicted Figure \ref{fig-plan-trid}\,(a). In the traditional variational framework, the action is obtained by summing this over the plane,
\[
S_{\Z^2} := \sum_{\vec n \in \Z^2} \cL_\trileg ( U(\vec n), U(\vec n + \vec e_i), U(\vec n + \vec e_j), U(\vec n + \vec e_i + \vec e_j), A_i, A_j ) ,
 \]
and the Euler--Lagrange equation, illustrated in Figure \ref{fig-plan-trid}\,(b), is obtained by varying the field $U$ at one lattice site:
\begin{align*}
 0 = \pdv{S_{\Z^2}}{U}
 &=\psi(U,U_i,A_i) - \psi(U,U_j,A_j) - \phi(U,U_{ij}, A_i-A_j) \\
 &\quad{} + \psi(U,U_{-i},A_i) - \psi(U,U_{-j},A_j) - \phi(U,U_{-i-j}, A_i-A_j) ,
\end{align*}
where the subscript $-i$ denotes a shift in the negative direction along the $i$-th coordinate axis. This Euler--Lagrange equation is the sum of two shifted instances of the quad equation, in the form \eqref{quad-u} and \eqref{quad-uij} respectively, with $\Theta = 0$.

\begin{figure}[t]
\centering
\begin{tikzpicture}[scale=3]
 \begin{scope}[scale=1]
 	\draw (0,0) node[below left] {$U$} -- (1,0) node[below right] {$U_i$} -- (1,1) node[above right] {$U_{ij}$} -- (0,1) node[above left] {$U_j$} -- cycle;
 \draw[very thick, blue] (0,0) -- node[below, black] {$L(U,U_i,A_i)$} (1,0);
 \draw[very thick, red] (0,1) -- node[above, rotate=90, black] {$-L(U,U_j,A_j)$} (0,0);
 \draw[very thick, red] (1,1) -- node[above, rotate=45, black] {\ $-\Lambda(U,U_{ij}, A_i-A_j)$} (0,0);
 \node at (.5,-.55) {(a)};
 \node at (0,0) {$\bullet$};
 \node at (1,0) {$\bullet$};
 \node at (0,1) {$\bullet$};
 \node at (1,1) {$\bullet$};
 \end{scope}
	\begin{scope}[shift=({3,.5}), scale=.7]
 	\draw (-1,1) -- (0,1) node[above] {$U_j$} -- (1,1) node[above right] {$U_{ij}$} -- (1,0) node[right] {$U_i$} -- (1,-1) -- (0,-1) node[below] {$U_{-j}$} -- (-1,-1) node[below left] {$U_{-i,-j}$} -- (-1,0) node[left] {$U_{-i}$} -- cycle;
 \node[below right] at (0,0) {$U$};
 \draw[very thick, blue] (0,0) -- (1,0);
 \draw[very thick, blue] (-1,0) -- (0,0);
 \draw[very thick, red] (0,1) -- (0,0);
 \draw[very thick, red] (0,0) -- (0,-1);
 \draw[very thick, red] (-1,-1) -- (0,0);
 \draw[very thick, red] (0,0) -- (1,1);
 \node at (0,-1.5) {(b)};
 \node at (0,0) {$\bullet$};
 \node at (1,0) {$\bullet$};
 \node at (0,1) {$\bullet$};
 \node at (-1,0) {$\bullet$};
 \node at (-1,-1) {$\bullet$};
 \node at (0,-1) {$\bullet$};
 \node at (1,1) {$\bullet$};
	\end{scope}
\end{tikzpicture}
\caption{(a) The stencil of the trident Lagrangian. (b) The discrete Euler--Lagrange equation involves three-leg structures on two squares.}
\label{fig-plan-trid}
\end{figure}
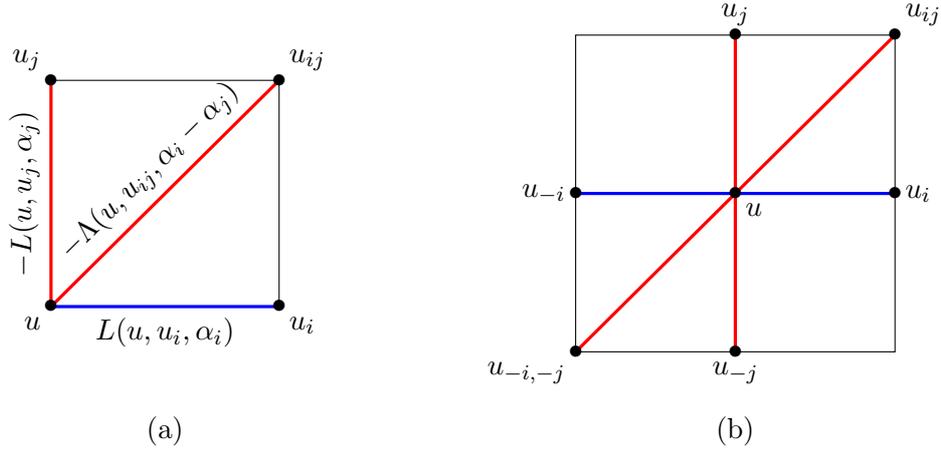

In order to get the general form of the three-leg equation (for H2, H3, A1$_{\delta=1}$, A2, Q1$_{\delta=1}$, Q2, Q3), we can add a term $2 \Theta U \pi {\rm i}$ at each lattice site and consider the action
\begin{align}
 S_{\Z^2}^{\vec \Theta}:={}& \sum_{\vec n \in \Z^2} \cL_\trileg ( U(\vec n), U(\vec n + \vec e_i), U(\vec n + \vec e_j), U(\vec n + \vec e_i + \vec e_j), A_i, A_j ) \nonumber\\
 &{} + \sum_{\vec n \in \Z^2} 2 \Theta(\vec n) U(\vec n) \pi {\rm i} .\label{planar-action-theta}
 \end{align}
Note that we do not get an Euler--Lagrange equation of the form \smash{$\pdv{S_{\Z^2}^{\vec\Theta}}{\Theta} = 0$} because we cannot take infinitesimal variations of the integer-valued $\Theta$. Note also that infinitesimal variations of~$U$, away from any branch cuts, do not induce any change in $\Theta$.

Alternatively, the functions $L$ and $\Lambda$ can be combined into a 3-point Lagrangian \cite{lobb2009lagrangian}
\begin{equation}
 \label{L_triang}
 \cL_\triang(U,U_i,U_j,A_i,A_j) := L(U,U_i,A_i) - L(U,U_j,A_j) - \Lambda(U_i,U_j, A_i-A_j) ,
\end{equation}
which we will refer to as the \emph{triangle Lagrangian}, inspired by the three-leg structure depicted Figure \ref{fig-triang}\,(a). The traditional Euler--Lagrange equation of $\cL_\triang$, depicted in Figure \ref{fig-triang}\,(b), is
\begin{align*}
 0 &= \pdv{}{U} ( \cL_\triang(U,U_i,U_j,A_i,A_j) + \cL_\triang(U_{-i},U,U_{-i,j},A_i,A_j) + \cL_\triang(U_{-j},U_{i,-j},U,A_i,A_j) ) \\
 &=\psi(U,U_i,A_i) - \psi(U,U_{-j},A_j) - \phi(U,U_{i,-j}, A_i-A_j) \\
 &\quad{} + \psi(U,U_{-i},A_i) - \psi(U,U_j,A_j) - \phi(U,U_{-i,j}, A_i-A_j) .
\end{align*}
This Euler--Lagrange equation is the sum of two shifted instances of the quad equation, in the form \eqref{quad-ui} and \eqref{quad-uj} respectively, with $\Theta = 0$. Again, to recover the equations with general values of $\Theta$, we can add $2 \Theta U \pi {\rm i}$ to the action at each lattice site.

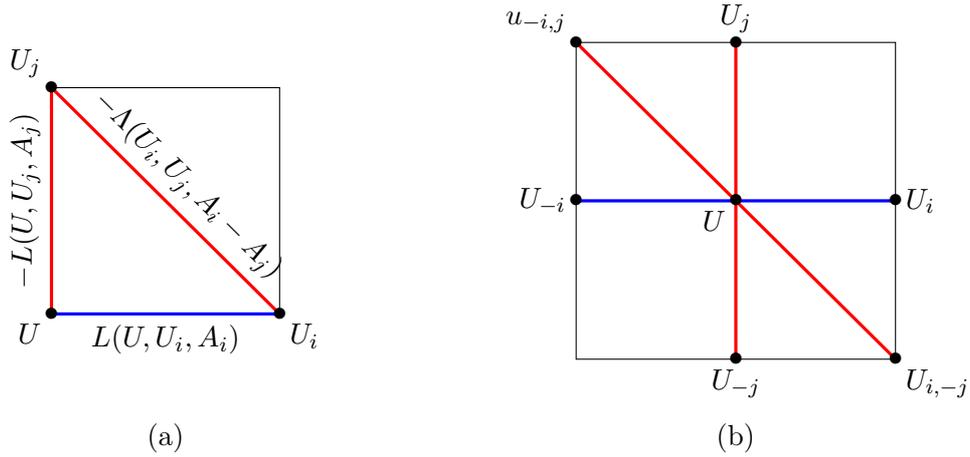
\begin{figure}[t]
\centering
\begin{tikzpicture}[scale=3]
	\draw (0,0) node[below left] {$U$} -- (1,0) node[below right] {$U_i$} -- (1,1) -- (0,1) node[above left] {$U_j$} -- cycle;
	\draw[very thick, blue] (0,0) -- node[below, black] {$L(U,U_i,A_i)$} (1,0);
	\draw[very thick, red] (0,1) -- node[above, rotate=90, black] {$-L(U,U_j,A_j)$} (0,0);
	\draw[very thick, red] (0,1) -- node[above, rotate=-45, black] {\ $-\Lambda(U_i,U_j, A_i-A_j)$} (1,0);
 \node at (.5,-.55) {(a)};
 \node at (0,0) {$\bullet$};
 \node at (1,0) {$\bullet$};
 \node at (0,1) {$\bullet$};

	\begin{scope}[shift=({3,.5}), scale=.7]
 	\draw (-1,1) node[above left] {$U_{-i,j}$} -- (0,1) node[above] {$U_j$} -- (1,1) -- (1,0) node[right] {$U_i$} -- (1,-1) node[below right] {$U_{i,-j}$} -- (0,-1) node[below] {$U_{-j}$} -- (-1,-1) -- (-1,0) node[left] {$U_{-i}$} -- cycle;
 \node[below left] at (0,0) {$U$};
 \draw[very thick, blue] (0,0) -- (1,0);
 \draw[very thick, blue] (-1,0) -- (0,0);
 \draw[very thick, red] (0,1) -- (0,0);
 \draw[very thick, red] (0,0) -- (0,-1);
 \draw[very thick, red] (-1,1) -- (0,0);
 \draw[very thick, red] (0,0) -- (1,-1);
 \node at (0,-1.5) {(b)};
 \node at (0,0) {$\bullet$};
 \node at (1,0) {$\bullet$};
 \node at (0,1) {$\bullet$};
 \node at (-1,0) {$\bullet$};
 \node at (-1,1) {$\bullet$};
 \node at (0,-1) {$\bullet$};
 \node at (1,-1) {$\bullet$};
	\end{scope}
\end{tikzpicture}
\caption{(a) The stencil of the triangle Lagrangian. (b) The discrete Euler--Lagrange equation involves three-leg structures on two squares.}
\label{fig-triang}
\end{figure}

\section{Lagrangian multiforms}
\label{sec-multiforms}

\subsection{General theory}
\label{sec-multiform-general}

In \cite{lobb2009lagrangian}, the crucial observation was made that the Lagrangian \eqref{L_triang} can be combined with the multidimensional consistency property, leading to the concept of Lagrangian multiforms. The key idea is to interpret a Lagrangian $\cL(U,U_i,U_j,U_{ij},A_i,A_j)$, which is skew-symmetric under interchange of indices $i \leftrightarrow j$, as a discrete 2-form on a higher-dimensional quadrilateral lattice. Here we will work on the lattice $\Z^N$, but more general lattices with quadrilateral faces can be considered. Instead of taking the action sum over a plane, Lagrangian multiform theory considers arbitrary 2-dimensional discrete surfaces in the lattice. Just like continuous a differential 2-form can be integrated over any oriented surface, a discrete 2-form can be integrated over an oriented discrete surface by summing over all elementary squares making up the surface. In Lagrangian multiform theory, we require that for every choice of discrete surface, the corresponding action sum is critical in two ways: with respect to variations of the fields $U,U_i,\ldots$ and with respect to variations of the discrete surface. An equivalent characterisation is to require that the action sum over every elementary cube $\{ \vec n + \mu_i \vec e_i + \mu_j \vec e_j + \mu_k \vec e_k \mid \mu_i,\mu_j,\mu_k \in \{0,1\} \} \subset \Z^N$ is equal to zero and critical with respect to variations of the fields \cite{boll2014integrability}.

Criticality with respect to variations of the fields leads to a set of generalised Euler--Lagrange equations. They are variously called \emph{multiform Euler--Lagrange equations}, \emph{multi-time Euler--Lagrange equations} (because in the continuous case the space of independent variables is spanned by time-variables of several commuting flows), or \emph{corner equations} \big(because they are derived from the corners of an elementary cube in $\Z^N$\big).
To derive the corner equations, consider the action sum over an elementary cube, depending on the fields $\vec U = (U,U_i,U_j,U_k,U_{ij},U_{jk},U_{ki})$ and lattice parameters $\vec A = (A_i,A_j,A_k)$,
\begin{align*}
 S = S(\vec U, \vec A) &= \cL(U_k,U_{ki},U_{jk}, U_{ijk},A_i,A_j) + \cL(U_i,U_{ij},U_{ki}, U_{ijk},A_j,A_k) \\
 &\quad{} + \cL(U_j,U_{jk},U_{ij},U_{ijk},A_k,A_i) - \cL(U,U_i,U_j,U_{ij},A_i,A_j) \\
 &\quad{} - \cL(U,U_j,U_k,U_{jk},A_j,A_k) - \cL(U,U_k,U_i,U_{ki},A_k,A_i) .
\end{align*}
This action is critical with respect to variations of the fields if and only if all of the following equations hold:
\begin{alignat*}{6}
 &\pdv{S}{U} = 0 ,\qquad && \pdv{S}{U_i} = 0 ,\qquad && \pdv{S}{U_j} = 0 ,\qquad && \pdv{S}{U_k} = 0 , && & \\
 & \qquad && \pdv{S}{U_{jk}} = 0 ,\qquad && \pdv{S}{U_{ki}} = 0 ,\qquad && \pdv{S}{U_{ij}} = 0 ,\qquad && \pdv{S}{U_{ijk}} = 0 .&
\end{alignat*}

For the cases in which the three-leg equations hold modulo $2 \pi {\rm i}$, similar to the planar action~\eqref{planar-action-theta}, we will consider
\begin{align*}
 S^{\vec \Theta} := S + 2 \pi {\rm i} ( \Theta U + \Theta_i U_i + \Theta_j U_j + \Theta_k U_k +\Theta_{ij} U_{ij} + \Theta_{jk} U_{jk} + \Theta_{ki} U_{ki} + \Theta_{ijk} U_{ijk}) .
\end{align*}
The action around the cube is symmetric under cyclic permutation of indices $(i,j,k)$. We will use the notation ${\!}+\cyclic$ to represent a sum over cyclic permutations. In addition, we will write $\sum 2 \Theta U \pi {\rm i}$ as a shorthand for the last eight terms. Thus, we can write
\begin{align*}
 S^{\vec \Theta} = \cL(U_k,U_{ki},U_{jk}, U_{ijk},A_i,A_j) - \cL(U,U_i,U_j,U_{ij},A_i,A_j) + \cyclic + \sum 2 \pi {\rm i} \Theta U
 .
\end{align*}
Note that the new terms in $S^\Theta$ are each associated to a vertex of the cube, not to a face. For this reason, we do not incorporate the $\Theta$-terms into the discrete Lagrangian 2-form. The price to pay is that the action now contains contributions of two different kinds: the sum $S$ over faces of the skew-symmetric function $\cL$, and the sum over vertices of the $\Theta$-dependent terms.

The action of a discrete Lagrangian 2-form over an elementary cube can be thought of as the discrete analogue to the exterior derivative. A property at the core of Lagrangian multiform theory is that, on solutions, the Lagrangian 2-form is closed, i.e., that the action over an elementary cube is zero:
\[ S = 0 . \]
This is the necessary and sufficient condition for the criticality of the action with respect to variations of the surface.

Similar to the three-leg formulation of the quad equations, the closure property depends on a careful choice of the branches of the logarithms and other functions in the definition of $\cL$. To achieve closure in a systematic way, we introduce additional integer fields to the action, one associated to each lattice parameter:
\begin{gather*}
 S^{\vec \Theta,\vec \Xi}:= \cL(U_k,U_{ki},U_{jk}, U_{ijk},A_i,A_j) - \cL(U,U_i,U_j,U_{ij},A_i,A_j) + \cyclic
 \\ \hphantom{S^{\vec \Theta,\vec \Xi}:=}{}
 + 2 \pi {\rm i} ( \Theta U + \Theta_i U_i + \Theta_j U_j + \Theta_k U_k +\Theta_{ij} U_{ij} + \Theta_{jk} U_{jk} + \Theta_{ki} U_{ki} + \Theta_{ijk} U_{ijk}) \\
 \hphantom{S^{\vec \Theta,\vec \Xi}:=}{}
 + 2 \pi {\rm i} (\Xi_i A_i + \Xi_j A_j + \Xi_k A_k).
\end{gather*}
For the equations with an additive three-leg form, the action does not contain the integers $\vec \Theta = (\Theta, \Theta_i, \ldots, \Theta_{ijk})$, but we will still need to include $\vec \Xi = (\Xi_i, \Xi_j, \Xi_k)$ to achieve closure. In~these cases, we consider
\begin{gather*}
 S^{\vec \Xi} := \cL(U_k,U_{ki},U_{jk}, U_{ijk},A_i,A_j) - \cL(U,U_i,U_j,U_{ij},A_i,A_j) + \cyclic \\
 \hphantom{S^{\vec \Xi} :=}{}
 + 2 \pi {\rm i} (\Xi_i A_i + \Xi_j A_j + \Xi_k A_k).
\end{gather*}

For equations H3, A2, and Q3, we will see below that the closure property only holds modulo~$4 \pi^2$, i.e., $S^{\vec \Theta,\vec \Xi}$ is a multiple of $4 \pi {\rm i}$. So the variational structure of these equations may be better described as a \emph{pluri-Lagrangian system}, which imposes the same variations of $U$, but does not require the closure condition.

More details on discrete Lagrangian multiforms and discrete pluri-Lagrangian systems can be found in \cite{bobenko2015discrete, boll2014integrability, hietarinta2016discrete, lobb2018variational}. However, one should note that these works assume that the right-hand sides of the three leg equations are always zero, instead of a multiple of $2 \pi {\rm i}$, and ignore similar subtleties when evaluating the action on the cube. Outside of Lagrangian multiform theory, integer multiples of $2 A_k \pi {\rm i}$ and $2 U \pi {\rm i}$ have been used to deal with branches of the complex logarithm in the context of star-triangle relations, see, for example, \cite{kels2021interaction, kels2023Twocomponent}.

In the next two subsections, we look at two ways of constructing Lagrangian multiforms for the ABS equation, each based on one of the planar Lagrangians of Section~\ref{subsec-planar}. The literature on discrete Lagrangian multiforms for quad equations almost exclusively considers multiforms based on the triangle Lagrangian, but we will argue that the trident Lagrangian is simpler to work with and show that it produces stronger Euler--Lagrange equations.
Then, in Section~\ref{subsec-clos}, we will give a rigorous proof of the closure (or almost closure) of both Lagrangian multiforms.

\subsection{Trident 2-form: quad equations are variational}
\label{sec-trident}

One of our main results is that the quad equations of the ABS list are variational, i.e., that they are equivalent to the corner equations of a Lagrangian multiform. This is in contrast to the Lagrangian multiforms for ABS equations that have previously been considered, which have corner equations that are slightly weaker than the quad equations \cite{bobenko2010lagrangian,boll2014integrability, lobb2009lagrangian,lobb2018variational}.
We will show that the trident Lagrangian $\cL_\trileg$, given in equation \eqref{L_trident} (and first considered in \cite{adler2003classification}) generalises to a~Lagrangian multiform with corner equations that produce the quad equations directly (in~their three-leg form).
The action over an elementary cube of $\cL_\trileg$, interpreted as a 2-form, is illustrated in Figure \ref{fig-trid}\,(b) and can be written as
\begin{align}
 S_\trileg & = \cL_\trileg(U_k,U_{ki},U_{jk},U_{ijk},A_i,A_j) - \cL_\trileg(U,U_i,U_j,U_{ij},A_i,A_j) + \cyclic \notag \\
 & = L(U_k,U_{ki},A_i) - L(U_k,U_{jk},A_j) - \Lambda(U_k,U_{ijk}, A_i-A_j) \notag \\
 &\quad{} - L(U,U_i,A_i) + L(U,U_j,A_j) + \Lambda(U,U_{ij}, A_i-A_j) + \cyclic \notag\\
 & = L(U_i,U_{ij},A_j) + \Lambda(U,U_{ij},A_i-A_j) -\smash{\reverse} + \cyclic , \label{Strid-sym}
\end{align}
where the notation \smash{$- \reverse$ }represents subtracting the point inverse of the preceding expression (i.e., replacing $U \leftrightarrow
U_{ijk}$, $U_i \leftrightarrow U_{jk}$, etc.), and $+\,\cyclic$ indicates the addition of terms obtained by cyclic permutation of $(i,j,k)$. Thus, equation \eqref{Strid-sym} is manifestly symmetric under cyclic permutations and skew-symmetric under point inversion. As in Section~\ref{sec-multiform-general}, we will extend the action with integer-valued fields $\Theta$ and $\Xi$:
\begin{align}
 S_\trileg^{\vec \Theta, \vec \Xi} &= S_\trileg + 2 \pi {\rm i} ( \Theta U + \Theta_i U_i + \Theta_j U_j + \Theta_k U_k +\Theta_{ij} U_{ij} + \Theta_{jk} U_{jk} + \Theta_{ki} U_{ki} + \Theta_{ijk} U_{ijk}) \nonumber\\
 &\quad{} + 2 \pi {\rm i} (\Xi_i A_i + \Xi_j A_j + \Xi_k A_k) \label{Strid-sym2}
 \end{align}
and
\[ S_\trileg^{\vec \Xi} = S_\trileg + 2 \pi {\rm i} (\Xi_i A_i + \Xi_j A_j + \Xi_k A_k). \]

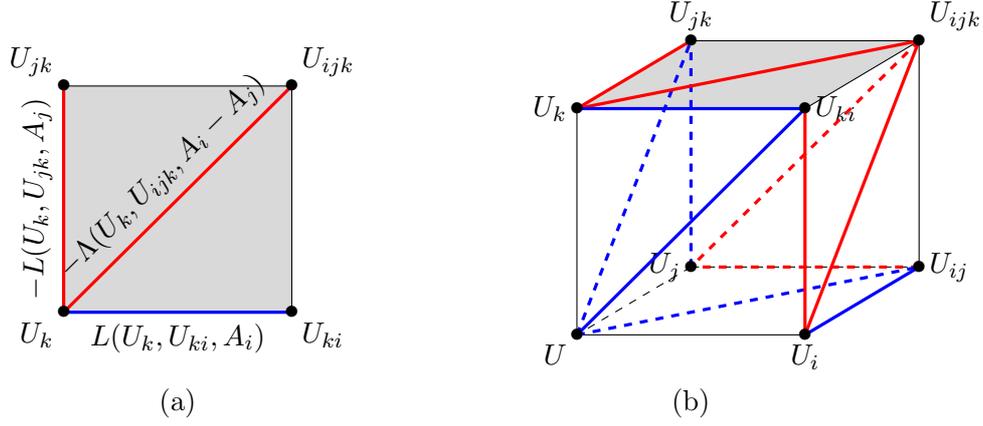
\begin{figure}[t]
\centering
\begin{tikzpicture}[scale=3]
	\begin{scope}[shift=({1.75,0})]
 \fill[gray!30] (0,0) -- (1,0) -- (1,1) -- (0,1) -- cycle;
		\draw (0,0) node[below left] {$U_k$} -- (1,0) node[below right] {$U_{ki}$} -- (1,1) node[above right] {$U_{ijk}$} -- (0,1) node[above left] {$U_{jk}$} -- cycle;
		\draw[very thick, blue] (0,0) -- node[below, black] {$L(U_k,U_{ki},A_i)$} (1,0);
		\draw[very thick, red] (0,1) -- node[above, rotate=90, black] {$-L(U_k,U_{jk},A_j)$} (0,0);
		\draw[very thick, red] (1,1) -- node[above, rotate=45, black] {\ $-\Lambda(U_k,U_{ijk}, A_i-A_j)$} (0,0);
 \node at (.5,-.4) {(a)};
 \node at (0,0) {$\bullet$};
 \node at (1,0) {$\bullet$};
 \node at (0,1) {$\bullet$};
 \node at (1,1) {$\bullet$};
	\end{scope}
	\begin{scope}[shift=({4,-.1}),scale=1, y={(5mm, 3mm)}, z={(0cm,1cm)}]
 \fill[gray!30] (0,0,1) -- (1,0,1) -- (1,1,1) -- (0,1,1) -- cycle;
		\draw[every edge/.append style={dashed}]
		(1,0,0) node[below] {$U_i$} -- (0,0,0) node[below left] {$U$} -- (0,0,1) node[left] {$U_k$} -- (0,1,1) node[above] {$U_{jk}$} -- (1,1,1) node[above right] {$U_{ijk}$} -- (1,1,0) node[right] {$U_{ij}$}-- cycle -- (1,0,1) node[right] {$U_{ki}$}-- (0,0,1)
		(1,0,1) -- (1,1,1)
		(0,1,0) edge (0,0,0) edge (0,1,1) edge (1,1,0) node[left] {$U_j$};
		
		\draw[very thick, red, dashed] (1,1,0) -- (0,1,0);
		\draw[very thick, blue, dashed] (0,1,0) -- (0,1,1);
		
		\draw[very thick, blue, dashed] (0,0,0) -- (0,1,1) ;
		\draw[very thick, blue] (0,0,0) -- (1,0,1) ;
		\draw[very thick, blue, dashed] (0,0,0) -- (1,1,0) ;
		
		\draw[very thick, red] (1,1,1) -- (0,0,1) ;
		\draw[very thick, red] (1,1,1) -- (1,0,0) ;
		\draw[very thick, red, dashed] (1,1,1) -- (0,1,0) ;
		
		\draw[very thick, blue] (1,0,0) -- (1,1,0);
		\draw[very thick, red] (0,1,1) -- (0,0,1);
		\draw[very thick, blue] (0,0,1) -- (1,0,1);
		\draw[very thick, red] (1,0,1) -- (1,0,0);
 \node at (1,-1) {(b)};
 \node at (0,0,0) {$\bullet$};
 \node at (1,0,0) {$\bullet$};
 \node at (0,1,0) {$\bullet$};
 \node at (0,0,1) {$\bullet$};
 \node at (1,1,0) {$\bullet$};
 \node at (0,1,1) {$\bullet$};
 \node at (1,0,1) {$\bullet$};
 \node at (1,1,1) {$\bullet$};
	\end{scope}	
\end{tikzpicture}

\caption[.]{(a) The leg structure of a single Lagrangian $\cL_\trileg(U_k,U_{ki},U_{jk},U_{ijk},A_i,A_j)$.
(b) The leg struc\-ture for the action on the cube of the trident 2-form $\cL_\trileg$.}\label{fig-trid}
\end{figure}

Consider the corner expression at $U_{ij}$,
\begin{align*}
 \pdv{S_\trileg^{\vec \Theta, \vec \Xi}}{U_{ij}} & = \pdv{}{U_{ij}} \left( L(U_i,U_{ij},A_j) - L(U_{ij},U_j,A_i) + \Lambda(U,U_{ij},A_i-A_j) + 2\Theta_{ij}U_{ij} \pi {\rm i} \right) \\
 &= \psi(U_{ij},U_i,A_j) - \psi(U_{ij},U_j,A_i) + \phi(U_{ij},U,A_i-A_j) + 2 \Theta_{ij} \pi {\rm i} \\
 &= -\mathcal{Q}^{(U_{ij})}_{ij} + 2\Theta_{ij} \pi {\rm i} .
\end{align*}
Setting this equal to zero gives us exactly the quad equation in three-leg form (see equation~\eqref{quad-uij}).
The expressions at $U_{jk}$ and $U_{ki}$ are analogous and follow from permutation of indices. At the~$U_k$ corner we find a similar expression, \smash{$\mathcal{Q}^{(U_{k})}_{ij} + 2\Theta_{k} \pi {\rm i}$}, which is what we would expect from the~skew-symmetry under point inversion. The expressions at $U_{i}$ and $U_{j}$ are analogous and follow from permutation of indices.

Let us consider the corner expression at $U$ produced by this 2-form,
\begin{align*}
 \pdv{S_\trileg^{\vec \Theta, \vec \Xi}}{U} & = \pdv{}{U} \left( \Lambda(U,U_{ij},A_i-A_j) + \Lambda(U,U_{jk},A_j-A_k + \Lambda(U,U_{ki},A_k-A_i) + 2 \Theta U \pi {\rm i} \right) \\
 &=\phi(U,U_{ij},A_i-A_j) + \phi(U,U_{jk},A_j-A_k) + \phi(U,U_{ki},A_k-A_i) + 2 \Theta \pi {\rm i} \\
 &= \mathcal{T}^{(U)} + 2 \Theta \pi {\rm i} .
\end{align*}
We find the tetrahedron equation in three-leg form (see equation \eqref{tetra-u}). Similarly, at the $U_{ijk}$ corner we find $-\mathcal{T}^{(U_{ijk})} + 2 \Theta_{ijk} \pi {\rm i}$.
Altogether, we conclude the following.

\begin{Theorem}
 The corner equations of the action \smash{$S_\trileg^{\vec \Theta, \vec \Xi}$} are the quad equations and tetrahedron equations in three-leg form:
 \begin{alignat*}{3}
 &\pdv{S_\trileg^{\vec \Theta, \vec \Xi}}{U} = 0 \implies \mathcal{T}^{(U)} = - 2 \Theta \pi {\rm i} , \qquad&&
 \pdv{S_\trileg^{\vec \Theta, \vec \Xi}}{U_{ijk}} = 0 \implies\mathcal{T}^{(U_{ijk})} = 2 \Theta_{ijk} \pi {\rm i} , & \\
 &\pdv{S_\trileg^{\vec \Theta, \vec \Xi}}{U_i} = 0 \implies \mathcal{Q}^{(U_i)}_{jk} = - 2 \Theta_{i} \pi {\rm i} ,\qquad&&
 \pdv{S_\trileg^{\vec \Theta, \vec \Xi}}{U_{jk}} = 0 \implies \mathcal{Q}^{(U_{jk})}_{jk} = 2 \Theta_{jk} \pi {\rm i} ,& \\
 &\pdv{S_\trileg^{\vec \Theta, \vec \Xi}}{U_j} = 0 \implies \mathcal{Q}^{(U_j)}_{ki} = - 2 \Theta_j \pi {\rm i} ,\qquad&&
 \pdv{S_\trileg^{\vec \Theta, \vec \Xi}}{U_{ki}} = 0 \implies \mathcal{Q}^{(U_{ki})}_{ki} = 2 \Theta_{ki} \pi {\rm i} ,& \\
 &\pdv{S_\trileg^{\vec \Theta, \vec \Xi}}{U_k} = 0 \implies \mathcal{Q}^{(U_k)}_{ij} = - 2 \Theta_k \pi {\rm i} ,\qquad&&
 \pdv{S_\trileg^{\vec \Theta, \vec \Xi}}{U_{ij}} = 0 \implies \mathcal{Q}^{(U_{ij})}_{ij} = 2 \Theta_{ij} \pi {\rm i} .&
 \end{alignat*}
 The corner equations of the action $S_\trileg^{\vec \Xi}$ are obtained form the above by setting $\vec \Theta = \vec 0$.
\end{Theorem}
\begin{proof}
 Take derivatives of \eqref{Strid-sym2} and recognise the three-leg quad expressions \eqref{quad-u}--\eqref{quad-uki}, as well as the three-leg tetrahedron expressions \eqref{tetra-u} and \eqref{tetra-uijk}.
\end{proof}

\begin{ExampleH1}
H1 is one of the equations for which the natural three-leg form is additive. This means that the functions $\phi$ and $\psi$ do not involve logarithms and that the multi-affine quad equation is equivalent to the three leg-forms \eqref{three-leg-0} with zero on the right-hand side. Hence, in this case we should not include the integer fields $\vec \Theta$ in the action. Furthermore, no variable transformation is needed, so we have $u = U$. Using the expressions \eqref{H1-L-Lambda} for $L$ and $\Lambda$, we find the trident Lagrangian
\begin{equation*}
 \cL_\trileg(u,u_i,u_j,u_{ij},\alpha_i-\alpha_j) = u u_i - u u_j - (\alpha_i - \alpha_j) \log(u - u_{ij}) .
\end{equation*}
The action around the elementary cube of the corresponding discrete 2-form is
\begin{align}
 S_\trileg
 & = u_k u_{ki} - u_k u_{jk} - (\alpha_i - \alpha_j) \log(u_k - u_{ijk}) \notag \\
 &\quad - u u_i + u u_j + (\alpha_i - \alpha_j) \log(u - u_{ij}) + \cyclic\notag\\
 & = u_i u_{ij} + (\alpha_i - \alpha_j) \log(u - u_{ij}) - \reverse + \cyclic .\label{H1-S-trident}
\end{align}
(We have also left out $\Xi_i$, $\Xi_j$, $\Xi_k$ here, because they do not affect the corner equations. However, they will be essential for the closure property in Section~\ref{subsec-clos}.)

The partial derivatives of the action \eqref{H1-S-trident} produce the three-leg forms of the tetrahedron and quad equations. For example, we have
\[
 \pdv{ S_\trileg}{u} = \frac{\alpha_i - \alpha_j}{u - u_{ij}} + \frac{\alpha_j - \alpha_k}{u - u_{jk}} + \frac{\alpha_k - \alpha_i}{u - u_{ki}}
 = \mathcal{T}^{(u)}
\]
and
\[
 \pdv{ S_\trileg}{u_{ij}} = u_i - u_j - \frac{\alpha_i - \alpha_j}{u - u_{ij}}
 = -\mathcal{Q}_{ij}^{(u_{ij})} .
\]
\end{ExampleH1}

\begin{ExampleH3}
For H3$_{\delta=0}$, we introduced the variable transformation $u = {\rm e}^U$ to find suitable using the expressions \eqref{H3-L-Lambda} for $L$ and $\Lambda$. These lead to the trident Lagrangian
\[ \cL_\trileg = -U U_i + U U_j + \dilog \bigl( {\rm e}^{A_i-A_j+U-U_{ij}} \bigr) - \dilog \bigl( {\rm e}^{-A_i+A_j+U-U_{ij}}\bigr) + (A_i-A_j) (U-U_{ij}) . \]
In this case, we obtained the additive three-leg form by taking the logarithm of the multiplicative one, so we need to include the integer fields $\Theta$ in the action:
\begin{align*}
 S_\trileg^{\vec \Theta, \vec \Xi} &= \bigl( -U_k U_{ki} + U_k U_{jk} \\
 &\hphantom{= \bigl(}{} + \dilog \bigl( {\rm e}^{A_i-A_j+U_k-U_{ijk}} \bigr) - \dilog \bigl( {\rm e}^{-A_i+A_j+U_k-U_{ijk}} \bigr) + (A_i-A_j) (U_k-U_{ijk}) \\
 &\hphantom{= \bigl(}{} - \dilog \bigl( {\rm e}^{A_i-A_j+U-U_{ij}} \bigr) + \dilog \bigl( {\rm e}^{-A_i+A_j+U-U_{ij}} \bigr) - (A_i-A_j) (U-U_{ij}) \bigr) + \cyclic \\
 &\quad + 2 \pi {\rm i} ( \Theta U + \Theta_i U_i + \Theta_j U_j + \Theta_k U_k +\Theta_{ij} U_{ij} + \Theta_{jk} U_{jk} + \Theta_{ki} U_{ki} + \Theta_{ijk} U_{ijk}) \\
 &\quad + 2 \pi {\rm i} (\Xi_i A_i + \Xi_j A_j + \Xi_k A_k) .
\end{align*}
In anticipation of the closure property (see Section~\ref{subsec-clos}), we have also included $\Xi_i$, $\Xi_j$, $\Xi_k$ here, but they play no role in the computation of the corner equations, which we turn to now.

Recall that
\[ \frac{\d}{\d z} \dilog({\rm e}^z) = - \log(1-{\rm e}^z) , \]
so as a corner expression, we have, for example,
\begin{align*}
 0 = \pdv{S_\trileg^{\vec \Theta, \vec \Xi}}{U_k}
 &{}= -U_{ki} + U_{jk} - \log \bigl(1 - {\rm e}^{A_i-A_j+U_k-U_{ijk}} \bigr) + \log \bigl( 1 - {\rm e}^{-A_i+A_j+U_k-U_{ijk}} \bigr) \\
 &\quad{} + (A_i - A_j) + 2 \Theta_k \pi {\rm i} .
\end{align*}
Taking the exponential, we find the equivalent equation
\[ {\rm e}^{U_{ki} - U_{jk} - A_i + A_j} \bigl( 1 - {\rm e}^{A_i-A_j+U_k-U_{ijk}} \bigr) = \bigl( 1 - {\rm e}^{-A_i+A_j+U_k-U_{ijk}} \bigr) .\]
Multiplying by ${\rm e}^{A_i+U_{jk}+U_{ijk}}$ and rearranging the terms, we recover the quad equation
\[ {\rm e}^{A_i}\bigl( {\rm e}^{U_k+U_{ki}} + {\rm e}^{U_{jk}+U_{ijk}} \bigr) - {\rm e}^{A_j} \bigl( {\rm e}^{U_k+U_{jk}} + {\rm e}^{U_{ki} + U_{ijk}} \bigr) = 0 . \]
\end{ExampleH3}

\subsection{Triangle 2-form}
\label{sec-triangle}

Next, we compare the above construction to the 2-form $\cL_\triang$, defined as in equation \eqref{L_triang}, which started the theory of Lagrangian multiforms \cite{lobb2009lagrangian}. We review that the action of this 2-form is critical on a set of equations weaker than the quad equations.

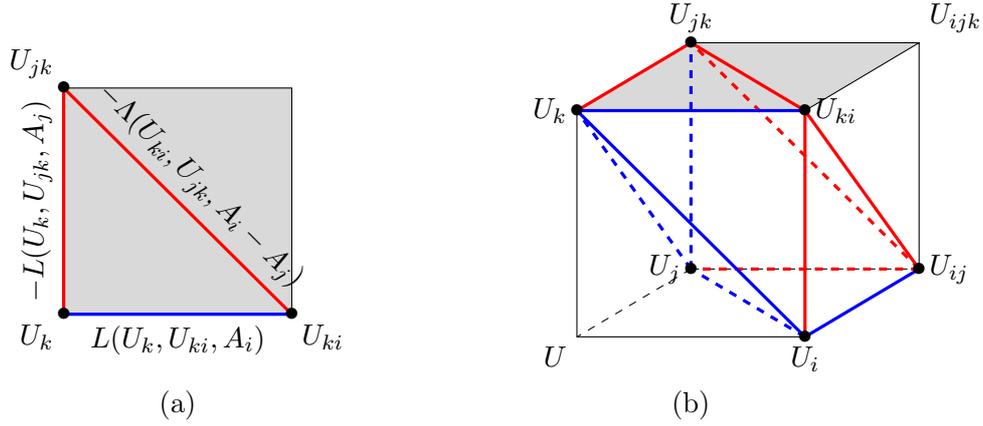
\begin{figure}[t]
\centering
\begin{tikzpicture}[scale=3]
 \begin{scope}[shift=({1.75,0})]
 \fill[gray!30] (0,0) -- (1,0) -- (1,1) -- (0,1) -- cycle;
		\draw (0,0) node[below left] {$U_k$} -- (1,0) node[below right] {$U_{ki}$} -- (1,1) -- (0,1) node[above left] {$U_{jk}$} -- cycle;
 	\draw[very thick, blue] (0,0) -- node[below, black] {$L(U_k,U_{ki},A_i)$} (1,0);
 	\draw[very thick, red] (0,1) -- node[above, rotate=90, black] {$-L(U_k,U_{jk},A_j)$} (0,0);
 	\draw[very thick, red] (0,1) -- node[above, rotate=-45, black] {\ $-\Lambda(U_{ki},U_{jk}, A_i-A_j)$} (1,0);
 \node at (.5,-.4) {(a)};
 \node at (0,0) {$\bullet$};
 \node at (1,0) {$\bullet$};
 \node at (0,1) {$\bullet$};
	\end{scope}
	\begin{scope}[shift=({4,-.1}),scale=1, y={(5mm, 3mm)}, z={(0cm,1cm)}]
 	\fill[gray!30] (0,0,1) -- (1,0,1) -- (1,1,1) -- (0,1,1) -- cycle;
 	\draw[every edge/.append style={dashed}]
 	(1,0,0) node[below] {$U_i$} -- (0,0,0) node[below left] {$U$} -- (0,0,1) node[left] {$U_k$} -- (0,1,1) node[above] {$U_{jk}$} -- (1,1,1) node[above right] {$U_{ijk}$} -- (1,1,0) node[right] {$U_{ij}$}-- cycle -- (1,0,1) node[right] {$U_{ki}$}-- (0,0,1)
 	(1,0,1) -- (1,1,1)
 	(0,1,0) edge (0,0,0) edge (0,1,1) edge (1,1,0) node[left] {$U_j$};
 	
 	\draw[very thick, red, dashed] (1,1,0) -- (0,1,1) ;
 	\draw[very thick, red] (0,1,1) -- (1,0,1) ;
 	\draw[very thick, red] (1,0,1) -- (1,1,0) ;
 	
 	\draw[very thick, blue] (1,0,0) -- (1,1,0);
 	\draw[very thick, red, dashed] (1,1,0) -- (0,1,0);
 	\draw[very thick, blue, dashed] (0,1,0) -- (0,1,1);
 	\draw[very thick, red] (0,1,1) -- (0,0,1);
 	\draw[very thick, blue] (0,0,1) -- (1,0,1);
 	\draw[very thick, red] (1,0,1) -- (1,0,0);
 	
 	\draw[very thick, blue, dashed] (1,0,0) -- (0,1,0) ;
 	\draw[very thick, blue, dashed] (0,1,0) -- (0,0,1) ;
 	\draw[very thick, blue] (0,0,1) -- (1,0,0) ;
 \node at (1,0,0) {$\bullet$};
 \node at (0,1,0) {$\bullet$};
 \node at (0,0,1) {$\bullet$};
 \node at (1,1,0) {$\bullet$};
 \node at (0,1,1) {$\bullet$};
 \node at (1,0,1) {$\bullet$};
 \node at (1,-1) {(b)};
 \end{scope}
 \end{tikzpicture}

 \caption[.]{(a) The leg structure of a single Lagrangian $\cL_\triang(U_k,U_{ki},U_{jk},A_i,A_j)$.
 (b) The leg structure for the action on an elementary cube of the triangle 2-form $\cL_\triang$ sits on an octahedral stencil.}
 \label{fig-triang-cube}
 \end{figure}

The action around the elementary cube of the triangle Lagrangian depends on an octahedral stencil depicted in Figure~\ref{fig-triang-cube},
\[
 S_\triang = L(U_i,U_{ij},A_j) + \Lambda(U_i,U_j,A_i-A_j) - \reverse + \cyclic .
\]
Since $U$ and $U_{ijk}$ do not appear in this action, there are no corner equations at these points.
The corner equation at $U_{ij}$ in terms of leg functions is
{\samepage\begin{align*}
	0 &= \pdv{S_\triang}{U_{ij}}
 \\
 &= \pdv{}{U_{ij}} \left( L(U_i,U_{ij},A_j) - L(U_j,U_{ij},A_i) - \Lambda(U_{ij},U_{ki}, A_j-A_k) - \Lambda(U_{jk},U_{ij}, A_k-A_i) \right) \\
 &= \psi(U_{ij},U_i,A_j) - \psi(U_{ij},U_j,A_i) - \phi(U_{ij},U_{ki}, A_j-A_k) - \phi(U_{ij},U_{jk}, A_k-A_i) \\
 & = \inversion{\mathcal Q}^{(U_{ij})}_{jk} + \inversion{\mathcal Q}^{(U_{ij})}_{ki} .
\end{align*}}
Similarly, at the corner $U_k$, we find the corner equation \smash{$-\mathcal{Q}^{(U_k)}_{jk} - \mathcal{Q}^{(U_k)}_{ki} = 0$}. The corner equations at $U_i$, $U_j$, $U_{jk}$, $U_{ki}$ are obtained from these equations by cyclic permutation of the indices.

For equations other than H1, Q1$_{\delta=0}$ and A1$_{\delta=0}$, solutions to the quad equations may not~satis\-fy \smash{$\mathcal Q^{(U_{ij})}_{jk} + \mathcal Q^{(U_{ij})}_{ki} = 0$} but rather \smash{$\mathcal Q\raisebox{-1pt}{${}^{(U_{ij})}_{jk}$} + \mathcal Q\raisebox{-1pt}{${}^{(U_{ij})}_{ki}$} = 2 \Theta_{ij} \pi {\rm i}$} for some $\Theta_{ij} \in \Z$. To account for~this, and analogous to the case of the trident Lagrangian, we add integer fields to the action of the triangle Lagrangian. Thus we obtain the following result.

\begin{Theorem}
 The corner equations of the action
 \begin{align*} S_\triang^{\vec \Theta, \vec \Xi} &= L(U_i,U_{ij},A_j) + \Lambda(U_i,U_j,A_i-A_j) - \reverse + \cyclic
 \\
 &\quad + 2 \pi {\rm i} ( \Theta U + \Theta_i U_i + \Theta_j U_j + \Theta_k U_k +\Theta_{ij} U_{ij} + \Theta_{jk} U_{jk} + \Theta_{ki} U_{ki} + \Theta_{ijk} U_{ijk}) \\
 &\quad + 2 \pi {\rm i} (\Xi_i A_i + \Xi_j A_j + \Xi_k A_k)
 \end{align*}
 are
 \begin{alignat*}{3}
 &\pdv{S_\triang^{\vec \Theta, \vec \Xi}}{U_i} = 0 \Rightarrow \mathcal{Q}^{(U_i)}_{ki} + \mathcal{Q}^{(U_i)}_{ij} = 2 \Theta_i \pi {\rm i} , \qquad&&
 \pdv{S_\triang^{\vec \Theta, \vec \Xi}}{U_{jk}} = 0 \Rightarrow \mathcal Q^{(U_{jk})}_{ki} + \mathcal Q^{(u_{jk})}_{ij} = - 2 \Theta_{jk} \pi {\rm i} ,& \\
 &\pdv{S_\triang^{\vec \Theta, \vec \Xi}}{U_j} = 0 \Rightarrow \mathcal{Q}^{(U_j)}_{ij} + \mathcal{Q}^{(U_j)}_{jk} = 2 \Theta_j \pi {\rm i} , \qquad&&
 \pdv{S_\triang^{\vec \Theta, \vec \Xi}}{U_{ki}} = 0 \Rightarrow \mathcal Q^{(U_{ki})}_{ij} + \mathcal Q^{(u_{ki})}_{jk} = - 2 \Theta_{ki} \pi {\rm i} ,& \\
 &\pdv{S_\triang^{\vec \Theta, \vec \Xi}}{U_k} = 0 \Rightarrow \mathcal{Q}^{(U_k)}_{jk} + \mathcal{Q}^{(U_k)}_{ki} = 2 \Theta_k \pi {\rm i} , \qquad&&
 \pdv{S_\triang^{\vec \Theta, \vec \Xi}}{U_{ij}} = 0 \Rightarrow \mathcal Q^{(U_{ij})}_{jk} + \mathcal Q^{(u_{ij})}_{ki} = - 2 \Theta_{ij} \pi {\rm i} ,&
 \end{alignat*}
 while \smash{$\pdv{S_\triang^{\vec \Theta, \vec \Xi}}{U} = 0$} and \smash{$\pdv{S_\triang^{\vec \Theta, \vec \Xi}}{U_{ijk}} = 0$} are trivially satisfied.

 The corner equations of the action $S_\triang^{\vec \Xi}$ with the trident $2$-form \eqref{L_triang} are obtained form the above by setting $\vec \Theta = \vec 0$.
\end{Theorem}
\begin{proof}
 These equations are obtained by comparing the partial derivatives of the action to the three-leg forms defined in \eqref{quad-u}--\eqref{quad-uki}.
\end{proof}

This 2-form produces a set of equations which vanish on the quad equations but are not equivalent to them because they lack the variables $U$ and $U_{ijk}$.

\begin{ExampleH1}
 For H1, we have
\[ \cL_\triang = u u_i - u u_j - (\alpha_i - \alpha_j) \log( u_i - u_j ) , \]
so{\samepage
\begin{align}
 S_\triang & = u_k u_{ki} - u_k u_{jk} - (\alpha_i - \alpha_j) \log( u_{ki} - u_{jk} ) + (\alpha_i - \alpha_j) \log( u_i - u_j ) + \cyclic , \nonumber\\
 & = u_i u_{ij} + (\alpha_i - \alpha_j) \log( u_i - u_j ) - \reverse + \cyclic .\label{H1-S-triangle}
\end{align}}%
As mentioned in the previous part of this example, for H1 we should not include the $\Theta$-terms in the action, because it naturally has an additive three leg form.

The corner equations at $u$ and $u_{ijk}$ are identically zero and the other six corner equations have a four-leg form. For example,
\begin{align*}
 0 = \pdv{S_\triang}{u_{ij}} &= u_i - u_j - \frac{\alpha_j - \alpha_k}{u_{ij} - u_{ki}} + \frac{\alpha_j - \alpha_k}{u_{jk} - u_{ij}} = \inversion{\mathcal Q}^{(u_{ij})}_{jk} + \inversion{\mathcal Q}^{(u_{ij})}_{ki}
\end{align*}
and
\begin{align*}
 0 = \pdv{S_\triang}{u_k} &= u_{ki} - u_{jk} + \frac{\alpha_k - \alpha_i}{u_k - u_i} - \frac{\alpha_j - \alpha_k}{u_j - u_k} = -\mathcal{Q}^{(u_{ij})}_{jk} - \mathcal{Q}^{(u_{ij})}_{ki} .
\end{align*}
Clearly, these equations are satisfied when the quad equations hold.
\end{ExampleH1}

\subsection{Closure relations}\label{subsec-clos}

The closure relation, i.e., the property that $S = 0$ on solutions, is a core idea of Lagrangian multiform theory. It was verified for the action $S_\triang$ of the triangle Lagrangian for H1, H2, H3, Q1, Q3$_{\delta=0}$, A1, and A2 by explicit calculations in \cite{lobb2009lagrangian} and a general proof was attempted in~\mbox{\cite{bobenko2010lagrangian, boll2014integrability}}. However, both these presentations ignore the fact that $L$ and $\Lambda$ have branch cuts. In the explicit computations, addition formulas for the (di)logarithm may introduce additive constants, depending on where in the complex plane their arguments lie. The proposed general argument assumes that the action depends continuously on the variables.

The following example shows that the closure relation does not always hold when the action is considered without the integer-valued fields $\vec \Xi$.

\begin{ExampleH1}
A particular solution to the quad equation
\[ (u_i - u_j)(u - u_{ij}) = \alpha_i - \alpha_j \]
with $\alpha_i = 1$, $\alpha_j = 2$, and $\alpha_k = 3$ is given by
\begin{alignat*}{6}
	& u = 0,\qquad	&& u_i = 1,\qquad && u_j = {\rm i},\qquad && u_k = 1+{\rm i},&& & \\
	& && u_{jk} = -1,\qquad && u_{ki} = 2 {\rm i},\qquad && u_{ij} = \frac12 + \frac12 {\rm i},\qquad && u_{ijk} = \frac65 + \frac35 {\rm i} .&
\end{alignat*}
For this solution, the actions \eqref{H1-S-trident} and \eqref{H1-S-triangle} take the values
\[ S_\trileg = -2\pi {\rm i} \qquad \text{and} \qquad S_\triang = 2\pi {\rm i} .\]
This shows that the closure property does not hold in general.

We can resolve this by adding $2 \Xi_{i} \alpha_i \pi {\rm i} + 2 \Xi_{j} \alpha_j \pi {\rm i} + 2 \Xi_{k} \alpha_k \pi {\rm i}$ to the action. Let us focus on the trident Lagrangian and set
\[ S_\trileg^{\vec \Xi} := S_\trileg + 2 \Xi_{i} \alpha_i \pi {\rm i} + 2 \Xi_{j} \alpha_j \pi {\rm i} + 2 \Xi_{k} \alpha_k \pi {\rm i} . \]
The integers $\Xi_{i}$, $\Xi_{j}$, $\Xi_{k}$ are to be chosen such that \smash{$\pdv{S_\trileg^{\vec \Xi}}{\alpha_i} = \pdv{S_\trileg^{\vec \Xi}}{\alpha_j} =\pdv{S_\trileg^{\vec \Xi}}{\alpha_k} = 0 $}, so
\begin{align*}
 & 2\Xi_{i} \pi {\rm i} = - \pdv{S_\trileg}{\alpha_i} = -\log(u-u_{ij}) + \log(u-u_{ki}) - \log(u_j-u_{ijk}) + \log(u_k-u_{ijk}) = 0 , \\
 & 2\Xi_{j} \pi {\rm i} = -\pdv{S_\trileg}{\alpha_j} = -\log(u-u_{jk}) + \log(u-u_{ij}) - \log(u_k-u_{ijk}) + \log(u_i-u_{ijk}) = -2\pi {\rm i} , \\
 & 2\Xi_{k} \pi {\rm i} = -\pdv{S_\trileg}{\alpha_k} = -\log(u-u_{ki}) + \log(u-u_{jk}) - \log(u_i-u_{ijk}) + \log(u_j-u_{ijk}) = 2\pi {\rm i} .
\end{align*}
This determines that $\Xi_i = 0$, $\Xi_j = -1$, $\Xi_k = 1$, and thus we recover that
\[ S_\trileg^{\vec \Xi} := S_\trileg + 2 \Xi_{i} \alpha_i \pi {\rm i} + 2 \Xi_{j} \alpha_j \pi {\rm i} + 2 \Xi_{k} \alpha_k \pi {\rm i}
= S_\trileg + 0 - 4 \pi {\rm i} + 6 \pi {\rm i} = 0.\]
\end{ExampleH1}

This will be our general strategy to achieve (almost) closure. In addition to an integer fields~$\Theta$ at each lattice site, we include one integer field $\Xi$ for each coordinate direction in the action. The values of the $\Xi$ are determined by variations of the lattice parameters. This is made precise in the following lemma, which can be considered as a 2-form version of the \emph{spectrality property} from \cite{suris2013variational}.

\begin{Lemma}
\label{lemma-dSdalpha}
 On solutions to the quad equations around the cube, there exists a choice of integers $\Xi_i,\Xi_j,\Xi_k \in \mathbb{Z}$ such that
 \[ \pdv{S^{\vec \Xi}_\trileg}{\alpha_i} = \pdv{S^{\vec \Theta,\vec \Xi}_\trileg}{\alpha_i} = 0
 \]
 and there exists a choice of $\Xi_i,\Xi_j,\Xi_k \in \mathbb{Z}$ such that
 \[
 \pdv{S^{\vec \Xi}_\triang}{\alpha_i} = \pdv{S^{\vec \Theta,\vec \Xi}_\triang}{\alpha_i} = 0 . \]
\end{Lemma}
\begin{proof}
 We could verify this by elementary but tedious case-by-case computations. Instead, we will prove this using the biquadratics that play a central role in the ABS classification. These are polynomials $h$ and $g$ satisfying
 \begin{align*}
 & h(u,u_i,\alpha_i)= \frac{1}{k(\alpha_i,\alpha_j)} \biggl( Q \pdv{^2 Q}{u_j \partial u_{ij}} - \pdv{Q}{u_j}\pdv{Q}{u_{ij}} \biggr) , \\
 & g(u,u_{ij},\alpha_i-\alpha_j)= \frac{1}{k(\alpha_i,\alpha_j)} \biggl( Q \pdv{^2 Q}{u_i \partial u_j} - \pdv{Q}{u_i}\pdv{Q}{u_j} \biggr) ,
 \end{align*}
 where $k$ is a skew-symmetric function of the lattice parameters, chosen such that $h$ only depends on the indicated lattice parameter. The biquadratics are related to the $L$ and $\Lambda$ by
 \begin{align}
	 &\pdv{L(U,U_i,A_i)}{A_i}\equiv \log h(u,u_i,\alpha_i) + \kappa(U) + \kappa(U_i) + c(A_i) \mod 2 \pi {\rm i} , \label{L-log(h)}\\
	 &\pdv{\Lambda(U,U_{ij},A_i-A_j)}{A_i}\equiv \log g(u,u_{ij},\alpha_i-\alpha_j) + \kappa(U) + \kappa(U_{ij}) - \gamma(A_i-A_j) \mod 2 \pi {\rm i} ,\nonumber
 \end{align}
 for some functions $\kappa$, $c$, $\gamma$, where the multi-affine variables $u$, $u_i$, $u_{ij}$ and parameters $A_i$, $A_j$ are considered as functions of the transformed variables $U$, $U_i$, $U_{ij}$ and parameters $\alpha_i$, $\alpha_j$. Equation~\eqref{L-log(h)} was proved by implicit means in \cite[Lemma 3]{bobenko2010lagrangian} and can be verified by explicit calculations using the expressions for $L$, $\Lambda$, $g$, $h$ listed in Appendix~\ref{AppendixA}.

 On solutions to a quad equation $Q_{ij} = 0$, the following biquadratic identities hold (see \cite[Lemma~1]{bobenko2010lagrangian} and \cite[Proposition 15]{adler2003classification}):
 \begin{align*}
	h(u,u_i,\alpha_i)h(u_{ij},u_j,\alpha_i) = h(u,u_j,\alpha_j)h(u_{ij},u_i,\alpha_j) = g(u,u_{ij},\alpha_i-\alpha_j)g(u_i,u_j,\alpha_i-\alpha_j) .
 \end{align*}
 Since the tetrahedron equation $T(u,u_{ij},u_{jk},u_{ki},\alpha_i,\alpha_j,\alpha_k) =0$ coincides with a quad equation of type $Q$, i.e., \smash{$T = Q^{(\text{type Q})}(u,u_{ij},u_{jk},u_{ki},\alpha_i-\alpha_j,-\alpha_j+\alpha_k)$}, it follows that on solutions to the tetrahedron equation there holds
 \begin{align*}
 \begin{split}
	g(u,u_{ij},\alpha_i-\alpha_j)g(u_{ki},u_{jk},\alpha_i-\alpha_j) & = g(u,u_{jk},\alpha_j-\alpha_k)g(u_{ki},u_{ij},\alpha_k-\alpha_i) \\
	& = g(u,u_{ki},\alpha_k-\alpha_i)g(u_{ij},u_{jk},\alpha_k-\alpha_i) .
\end{split}
 \end{align*}
 Similarly, on solutions to the tetrahedron equation $T(u_{ijk},u_{k},u_{i},u_{j},\alpha_i,\alpha_j,\alpha_k) = 0$ we have
 \begin{align*}
	g(u_{ijk},u_{k},\alpha_i-\alpha_j)g(u_{j},u_{i},\alpha_i-\alpha_j) & = g(u_{ijk},u_{i},\alpha_j-\alpha_k)g(u_{j},u_{k},\alpha_k-\alpha_i) \\
	& = g(u_{ijk},u_{j},\alpha_k-\alpha_i)g(u_{k},u_{i},\alpha_k-\alpha_i) .
 \end{align*}

 Hence, on solutions to the quad equations around the cube (which imply the tetrahedron equations), we compute
 \begin{align*}
	\pdv{S_\trileg^{\vec \Theta,\vec \Xi}}{\alpha_i} &=
	\pdv{}{\alpha_i} \bigl( L(U_k,U_{ki},\alpha_i) - L(U_j,U_{ij},\alpha_i) + \Lambda(U,U_{ij},\alpha_i-\alpha_j) + \Lambda(U,U_{ki},\alpha_k-\alpha_i)\\
	&\hphantom{=\pdv{}{\alpha_i} \bigl(}{}\, - \Lambda(U_{ijk}, U_{k}, \alpha_i-\alpha_j) - \Lambda(U_{ijk},U_{j},\alpha_k-\alpha_i) \bigr) + 2 \Xi_i \pi {\rm i}\\
	&\equiv \log\biggl( \frac{h(u_k,u_{ki},\alpha_i)g(u,u_{ij},\alpha_i-\alpha_j)g(u_{ijk},u_{j},\alpha_k-\alpha_i)}{h(u_j,u_{ij},\alpha_i)g(u,u_{ki},\alpha_k-\alpha_i)g(u_{ijk}, u_{k}, \alpha_i-\alpha_j)} \biggr) \mod 2 \pi {\rm i} \\
	&\equiv \log\biggl( \frac{g(u, u_{ij}, \alpha_i-\alpha_j) g(u_i,u_j,\alpha_i-\alpha_j)}{h(u,u_{i},\alpha_i)h(u_j,u_{ij},\alpha_i)} \biggr) \\
 &\quad + \log\biggl( \frac{h(u,u_{i},\alpha_i)h(u_k,u_{ki},\alpha_i)}{g(u, u_{ki}, \alpha_k-\alpha_i)g(u_k,u_i,\alpha_k-\alpha_i)} \biggr) \\
	& \quad + \log\biggl( \frac{g(u_{ijk},u_{j},\alpha_k-\alpha_i)g(u_{k},u_{i},\alpha_k-\alpha_i)}{g(u_{ijk},u_{k},\alpha_i-\alpha_j)g(u_{j},u_{i},\alpha_i-\alpha_j) } \biggr) \mod 2 \pi {\rm i} \\
 &\equiv 0 \mod 2 \pi {\rm i} ,
 \end{align*}
 so there is a $\Xi_i \in \Z$ such that $\pdv{S^{\Theta,\Xi}_\trileg}{\alpha_i} = 0$.
 The proof for $S_\triang^{\vec \Theta,\vec \Xi}$ is analogous.
\end{proof}

As a function of the fields and parameters, $S^{\vec \Theta, \vec \Xi}$ is a smooth function, except where its constituent functions have branch cuts. The following lemma states that $S^{\vec \Theta, \vec \Xi}$ is constant on solutions away from branch cuts, and gives the possible values for the jump when crossing a~branch cut.

\begin{Lemma}\label{lemma-Sjump}
 Consider a smooth $1$-parameter family of solutions $\vec U(t)$ to a transformed quad equation with lattice parameters $\vec A(t)$. Let \smash{$S^{\vec \Theta, \vec \Xi}$} denote the action on the cube of either the trident or the triangle Lagrangian.

 For ${\rm H}1$, ${\rm A}1_{\delta=0}$, ${\rm Q}1_{\delta=0}$, we set $\vec \Theta(t) = \vec 0$, so that $S^{\vec \Theta, \vec \Xi} = S^{\vec \Xi}$. For the other equations, $\vec \Theta(t)$ is uniquely determined by the corner equations. For all equations, $\vec \Xi(t)$ is determined by the~conditions \smash{$\pdv{S^{\vec \Theta, \vec \Xi}}{\vec A} = \vec 0$}.

 Assume $S^{\vec \Theta, \vec \Xi}$ has a discontinuity at $(\vec U(t_0), \vec A(t_0), \vec \Theta(t_0), \vec \Xi(t_0) )$ and denote the corresponding jump in the value of $S^{\vec \Theta, \vec \Xi}$ by
 \begin{gather*}
 J\bigl[S^{\vec \Theta, \vec \Xi}\bigr] := \lim_{\tau \to 0^+} \bigl( S^{\vec \Theta, \vec \Xi}(\vec U(t_0+\tau), \vec A(t_0+\tau), \vec \Theta(t_0+\tau), \vec \Xi(t_0+\tau)) \\
 \hphantom{J[S^{\vec \Theta, \vec \Xi}] := \lim_{\tau \to 0^+} \bigl(}{}\
 - S^{\vec \Theta, \vec \Xi}(\vec U(t_0-\tau), \vec A(t_0+\tau), \vec \Theta(t_0-\tau), \vec \Xi(t_0-\tau)) \bigr) .
 \end{gather*}
 Then
 \begin{itemize}\itemsep=0pt\samepage
 \item for ${\rm H}1$, ${\rm H}2$, ${\rm Q}1$, ${\rm Q}2$, ${\rm A}1$, we have $J\bigl[S^{\vec \Theta, \vec \Xi}\bigr] = 0$,

 \item for ${\rm H}3$, ${\rm Q}3$, ${\rm A}2$, we have $J\bigl[S^{\vec \Theta, \vec \Xi}\bigr] \equiv 0 \mod 4 \pi^2$.
 \end{itemize}
\end{Lemma}
\begin{proof}
 We show that the change in $\vec \Theta$ and $\vec \Xi$ cancels any branch jump in the (di)logarithm functions. Consider a 1-parameter family of solutions $(\vec U(t), \vec A(t), \vec \Theta(t), \vec \Xi(t) )$ such that at $t=t_0$ a branch cut is crossed in a term of $S^{\vec \Theta, \vec \Xi}$. There are three possible forms for this term:
 \begin{enumerate}\itemsep=0pt
 \item H1, Q1$_{\delta=0}$, A1$_{\delta=0}$ contain terms of the form
 \[ \biggl( \sum_n \varepsilon_n A_n \biggr) \log \biggl( \sum_m \delta_m U_m \biggr), \] where $n$ ranges over a subset of $\{i,j,k\}$, $m$ ranges over a subset of the corners of the cube, and $\varepsilon_n,\delta_m \in \{-2,-1,1,2\}$.
 The jump in the logarithm could be either~$2\pi {\rm i}$ or~$-2 \pi {\rm i}$, depending on the direction in which the branch cut is crossed. Assume without loss of generality that this is towards the positive imaginary direction and that $(\sum_n \varepsilon_n A_n) \log(\sum_m \delta_m U_m)$ occurs with + sign in the action. Then
 \begin{align*}
J\biggl[ \log \biggl(\sum_m \delta_m U_m\biggr) \biggr] &= \lim_{\tau \to 0^+} \biggl( \log \biggl( \sum_m \delta_m U_m(t_0+\tau) \biggr) - \log \biggl(\sum_m \delta_m U_m(t_0-\tau) \biggr) \biggr) \\
 &= 2 \pi {\rm i}
 \end{align*}
 and
 \[ J[S] = 2 \pi {\rm i} \sum_n \varepsilon_n A_n ,\]
 so
 \[ J\biggl[ \pdv{S}{A_n} \biggr] = 2 \pi {\rm i} \varepsilon_n . \]
 Since \smash{$\pdv{S^{\vec \Theta, \vec \Xi}}{A_n} = 0$} for all solutions, this implies that
 \[ J[\Xi_n] = - \varepsilon_n . \]
 Hence,
 \begin{align*}
 J\bigl[ S^{\vec \Theta, \vec \Xi} \bigr]
 &= J[S] + 2 \pi {\rm i} \sum_n A_n J[\Xi_n]
 = 0 .
 \end{align*}

 \item H2, Q1$_{\delta=1}$, Q2, A1$_{\delta=1}$ contain terms of the form
 \[ \biggl( \sum_n \varepsilon_n X_n \biggr) \log \biggl( \sum_n \varepsilon_n X_n \biggr), \]
 where $\varepsilon_n \in \{1, -1\}$ and the $X_n$ are a subset of the lattice parameters and the fields on the corners of the cube. Note that the two sums are identical.
 Assume without loss of generality that the term $( \sum_n \varepsilon_n X_n ) \log ( \sum_n \varepsilon_n X_n )$ occurs with a~+~sign in the action and that the branch cut is crossed from the lower half of the complex plane into the upper half, then
 \[ J\biggl[\log \biggl( \sum_n \varepsilon_n X_n \biggr) \biggr] = 2 \pi {\rm i} . \]
 Then
 \[ J[S] = 2 \pi {\rm i} \sum_n \varepsilon_n X_n \]
 and
 \[ J\biggl[ \pdv{S}{X_n} \biggr] = 2 \pi {\rm i} \varepsilon_n . \]
 So, depending on whether $X_n$ is a lattice parameter or a field, we have either
 \[ J[ \Xi_n ] = - \varepsilon_n \qquad \text{or} \qquad J[ \Theta_n ] = - \varepsilon_n . \]
 Hence,
 \begin{align*}
 J\bigl[ S^{\vec \Theta, \vec \Xi} \bigr]
 &= J[S] + \sum_{n\colon X_n \text{ is lattice parameter}} 2 \pi {\rm i} X_n J[\Xi_n] + \sum_{n\colon X_n \text{ is field}} 2 \pi {\rm i} X_n J[\Theta_n]
 = 0 .
 \end{align*}

 \item H3, Q3, A2 contain terms of the form
 \[ \dilog \biggl( \exp \biggl( \sum_n \delta_n X_n \biggr) \biggr), \]
 where $\varepsilon_n \in \{1, -1\}$ and the $X_n$ are a subset of the lattice parameters and the fields on the corners of the cube.
 Note that the dilogarithm function $\dilog(z)$ has a branch jump at $z \in (1,\infty) \subset \R$, where the value jumps by $\pm 2 \pi {\rm i} \log(z)$. Assume that branch cut is crossed from the lower half plane into the upper half plane and that the term $\dilog ( \exp ( \sum_n \varepsilon_n X_n ) )$ occurs with a + sign in the action, then the branch jump satisfies
 \begin{align*}
 J[S] = J\biggl[ \dilog \biggl(\exp \biggl( \sum_n \varepsilon_n X_n \biggr) \biggr) \biggr]
 &\equiv 2 \pi {\rm i} \sum_n \varepsilon_n X_n \mod 4 \pi^2 .
 \end{align*}
 Then
 \[ J\biggl[ \pdv{S}{X_n} \biggr] =
 J\biggl[ - \varepsilon_n \log \biggl( 1-\exp \biggl(\sum_m \varepsilon_m X_m \biggr) \biggr) \biggr] = 2 \pi {\rm i} \varepsilon_n . \]
 So, depending on whether $X_n$ is a lattice parameter or a field, we have either
 \[ J[ \Xi_n ] = -\varepsilon_n \qquad \text{or} \qquad J[ \Theta_n ] = -\varepsilon_n . \]
 Hence,
 \begin{align*}
 J[ S^{\vec \Theta, \vec \Xi} ]
 &= J[S] + \sum_{n \colon X_n \text{ is lattice parameter}} 2 \pi {\rm i} X_n J[\Xi_n] + \sum_{n \colon X_n \text{ is field}} 2 \pi {\rm i} X_n J[\Theta_n] \\
 &\equiv 0 \mod 4 \pi^2 .
 \tag*{\qed}
 \end{align*}
 \end{enumerate}
 \renewcommand{\qed}{}
\end{proof}

\begin{Theorem}
 \label{thm-closed}
 Let either \smash{$S = S_\trileg^{\vec \Theta, \vec \Xi}$} or \smash{$S = S_\triang^{\vec \Theta, \vec \Xi}$}.
 For any solution $\vec U$ to the $($transformed$)$ quad equations, with lattice parameters $\vec A$ and associated integer fields $\vec \Theta$, $\vec \Xi$, there holds
 \[ S(\vec U, \vec A, \vec \Theta, \vec \Xi) =
 \begin{cases}
 0 & \text{for ${\rm H}1$, ${\rm H}2$, ${\rm Q}1$, ${\rm Q}2$, and ${\rm A}1$,} \\
 4 k \pi^2, \ k \in \Z, & \text{for ${\rm H}3$, ${\rm Q}3$, and ${\rm A}2$.}
 \end{cases}\]
 Note that for ${\rm H}1$, ${\rm A}1_{\delta=0}$, ${\rm Q}1_{\delta=0}$, we have $\vec \Theta = \vec 0$.
\end{Theorem}

\begin{proof}[Proof for H1, H2, Q1, Q2, A1]
 Given a solution $(\vec U, \vec A) = (U,U_i,\ldots,U_{ijk},A_i,A_j,A_k)$, consider the 1-parameter family of scaled solutions $\vec V(t) = \bigl(t \vec U, t^2 \vec A\bigr)$ in case of H1 and $\vec V(t) = (t \vec U, t \vec A)$ for the other equations. For $t>0$, this family does not encounter any singularities or branch cuts, because the complex argument of the argument of every logarithm function in the action is independent of $t$. For the same reasons, the integer fields $\vec \Xi$ and $\vec \Theta$ associated to this family of solutions are independent of $t$.

 The corner equations
 \begin{alignat*}{6}
 &\pdv{S}{U} = 0 ,\qquad && \pdv{S}{U_i} = 0 ,\qquad && \pdv{S}{U_j} = 0 ,\qquad && \pdv{S}{U_k} = 0 ,&& & \\
 & && \pdv{S}{U_{jk}} = 0 ,\qquad && \pdv{S}{U_{ki}} = 0 ,\qquad && \pdv{S}{U_{ij}} = 0 ,\qquad && \pdv{S}{U_{ijk}} = 0 , &
 \end{alignat*}
 together with the equations from Lemma~\ref{lemma-dSdalpha}
 \begin{align*}
 &\pdv{S}{A_i} = 0 ,\qquad  \pdv{S}{A_j} = 0 ,\qquad \pdv{S}{A_k} = 0
 \end{align*}
 imply that the gradient of $S$ vanishes. Hence, $S(\vec V(t), \vec \Theta, \vec \Xi)$ is constant for $t>0$. Furthermore, elementary calculus shows that in each of these cases, $\lim_{t \to 0} S(\vec V(t), \vec \Theta, \vec \Xi) = 0$, so we conclude that $S(\vec V(t), \vec \Theta, \vec \Xi) = 0$ for all $t>0$. In particular, we have \[ S(\vec U, \vec A, \vec \Theta, \vec \Xi) = S(\vec V(1), \vec \Theta, \vec \Xi) = 0 . \tag*{\qed} \]
 \renewcommand{\qed}{}
\end{proof}

In the remaining cases, we will not be able to give such an explicit description of a suitable one-parameter family of solutions, but the following lemma allows us to establish the existence of such a family.

\begin{Lemma}
 \label{lemma-1+Ot}
 Let $w, w_i, w_j, \gamma_i, \gamma_j \in \C$ with $w_i \neq w_j$ and consider the following initial data, parametrised by $t \in \R$:
 \begin{alignat*}{4}
 & v(t) = 1 + w t ,\qquad
 && v_i(t) = 1 + w_i t ,\qquad
 && v_j(t) = 1 + w_j t ,& \\
 & &&\beta_i(t) = 1 + \gamma_i t^2 ,\qquad
 &&\beta_j(t) = 1 + \gamma_j t^2 .&
 \end{alignat*}
 Then the solution to the multi-affine quad equation $Q(v,v_i,v_j,v_{ij},\beta_i,\beta_j) = 0$ satisfies $v_{ij}(t) = 1 + O(t)$.
\end{Lemma}
\begin{proof}
This follows from elementary case-by-case computations.
\end{proof}

\begin{proof}[Proof of Theorem \ref{thm-closed} for H3, Q3, and A2]
 Given a solution $(\vec U, \vec A)$ to the transformed equations, and the corresponding solution $(\vec u, \vec \alpha)$ to the multi-affine equations. Consider a~family $(\vec v(t), \vec \beta(t))$ of solutions to the multi-affine equations such that
 \begin{alignat*}{5}
 & v(t) = 1 + w t ,\qquad
 && v_i(t) = 1 + w_i t ,\qquad
 && v_j(t) = 1 + w_j t ,\qquad
 && v_k(t) = 1 + w_k t , &\\
 &\beta_i(t) = 1 + \gamma_i t^2 ,\qquad
 &&\beta_j(t) = 1 + \gamma_j t^2 ,\qquad
 &&\beta_k(t) = 1 + \gamma_k t^2 ,&
 \end{alignat*}
 where the constants $w$, $w_i$, $w_j$, $w_k$, $\gamma_i$, $\gamma_j$, $\gamma_k$ are chosen such that $(\vec v(1), \vec \beta(1)) = (\vec u, \vec \alpha)$.

 We assume that we are dealing with a generic solution $(\vec U, \vec A)$, in the sense that~$w_i$,~$w_j$,~$w_k$ are distinct. If not, we can consider a small perturbation of the original solution.

 Applying Lemma~\ref{lemma-1+Ot} to the faces of the cube adjacent to the vertex of $v$, we obtain that $v_{ij} = 1+ O(t)$, $v_{jk} = 1+ O(t)$, and $v_{ki} = 1+ O(t)$. Applying Lemma~\ref{lemma-1+Ot} to the tetrahedron equation, which is always a quad equation of type Q, we find that $v_{ijk} = 1 + O(t)$. Now consider a smooth one-parameter family $(\vec V(t), \vec B(t))$ of solutions to the transformed equations, corresponding to $(\vec v(t), \vec \beta(t))$ at each $t$. Let $\vec \Theta(t)$, $\vec \Xi(t)$ be the integer fields associated to this solution.

 At each of the lattice sites, we have either $v(t) = {\rm e}^{V(t)}$ or $v(t) = \cosh(V(t))$; in both cases it follows that ${\rm e}^{V(t)} = 1 + O(t)$ and hence $\lim_{t \to 0} V(t) \equiv 0 \mod 2 \pi {\rm i}$. By inspection of the functions~$L$ and~$\Lambda$ given in Appendix~\ref{AppendixA}, and noting that $2 \pi {\rm i} \lim_{t \to 0} V(t) \Theta(t) \equiv 0 \mod 4 \pi^2$, it now follows that
 \[ \lim_{t \to 0} S(\vec V(t), \vec B(t), \vec \Theta(t), \vec \Xi(t)) \equiv 0 \mod 4 \pi^2. \]
If $S$ is differentiable, the corner equations and Lemma~\ref{lemma-dSdalpha} mean that the gradient of $S$ vanishes. If $S$ is not differentiable, it encounters a branch discontinuity, where it jumps by a multiple of~$4 \pi^2$ as shown in Lemma~\ref{lemma-Sjump}. This shows that
 \begin{align*}
 S(\vec U, \vec A,\vec \Theta, \vec \Xi) &= S(\vec V(1), \vec B(1), \vec \Theta(1), \vec \Xi(1)) \\
 &\equiv \lim_{t \to 0} S(\vec V(t), \vec B(t), \vec \Theta(t), \vec \Xi(t))
 \equiv 0 \mod 4 \pi^2 . \tag*{\qed}
 \end{align*}
 \renewcommand{\qed}{}
\end{proof}

\begin{ExampleH3}
 To see that the statement of Theorem \ref{thm-closed} is as strong as can be hoped for, we numerically computed a number of solutions to the equations of the ABS list and readily observed examples for H3, A2, and Q3, where the action is a non-zero multiple of $4 \pi^2$. For example, for H3$_{\delta = 0}$ with initial data
 \begin{alignat*}{5}
 & A_i = 2 - 2{\rm i} ,\qquad && A_j = 1 + 2{\rm i} ,\qquad && A_k = 1 , \qquad &&&\\
 &U = 1 - {\rm i} ,\qquad && U_i = 1 + {\rm i} ,\qquad && U_j = {\rm i} ,\qquad && U_k = 2 + {\rm i} ,&
 \end{alignat*}
 we find a (non-unique) solution
 \begin{align*}
 &U_{ij} = 1.49110840612501 + 1.80621078461947{\rm i} , \\
 &U_{jk} = 1.00000000000000 + 3.05123312020969{\rm i} , \\
 &U_{ik} = 0.459012522305983 + 1.68677242027984{\rm i} , \\
 &U_{ijk} = 1.05436866494958 - 2.73523350791485{\rm i} .
 \end{align*}
 for which the non-zero integer fields are $\Theta = -1$, $\Theta_k = 1$, $\Theta_{ij} = 1$, $\Theta_{ijk} = -1$ and the action is computed as
 \[ S_\trileg^{\vec \Theta, \vec \Xi} = -39.4784176043574 - 1.8 \cdot 10^{-15} {\rm i} . \]
This agrees with $-4 \pi^2$ up to nine decimal places, suggesting that the difference is due to rounding errors.
\end{ExampleH3}

\subsection{Arbitrary surfaces}
\label{sec-arbitrary}

So far, we have always considered the action on a single elementary cube. We have used the strategy of adding $\vec \Theta = (\Theta, \Theta_i,\ldots, \Theta_{ijk})$ and $\vec \Xi = (\Xi_i,\Xi_j,\Xi_k)$ to achieve equivalence of the three-leg forms to the quad equations and the (almost-)closure property in that setting. A~natural generalisation to the whole lattice $\Z^N$ would be to assign a $\Theta$ to each lattice site and to have $N$ globally defined $\Xi$s. However, if we fix the integer fields first and then choose a cube to work with, we will run into contradictions. There is no reason why different cubes should need the same $\Xi$s, or why neighbouring cubes should need the same value of $\Theta$ at their common vertices.

\begin{ExampleH2}
 For H2, with $\alpha_i = 1-{\rm i}$, $\alpha_j= -{\rm i}$, $\alpha_k = 1$, we compute fields in two adjacent cubes. With initial conditions
 \[ u = 0, \qquad u_i = 1, \qquad u_j = 1+{\rm i}, \qquad u_k = {\rm i}, \]
 we find
 \[u_{ij} = -2-{\rm i}, \qquad u_{jk} = -3+{\rm i}, \qquad u_{ik} = - \frac65 + \frac35 {\rm i}, \qquad u_{ijk} = \frac65 + \frac35 {\rm i} . \]
 In the cube $(u,u_i,u_j,u_k,u_{ij},u_{jk},u_{ki},u_{ijk})$, this induces the integer fields
 \begin{alignat}{6}
 &\Xi_i = 0 , \qquad && \Xi_j = -1 , \qquad&& \Xi_k = 1 , && && & \nonumber\\
 &\Theta = -1 , \qquad && \Theta_i = 1 , \qquad &&\Theta_j = 0 , \qquad &&\Theta_k = 1 , && &\nonumber\\
 & && \Theta_{ij} = 0 , \qquad&& \Theta_{jk} = -1 , \qquad&& \Theta_{ki} = -1 , \qquad&& \Theta_{ijk} = 1 . & \label{intfields1}
 \end{alignat}
 To extend this to the cube $(u_i,u_{ii},u_{ij},u_{ki},u_{iij},u_{ijk},u_{kii},u_{iijk})$, we choose the initial value $u_{ii} = 1$ and find
 \[ u_{iij} = \frac{12}{17} + \frac{14}{17}{\rm i} ,
 \qquad u_{iik} = \frac{11}{37} + \frac{45}{37}{\rm i} ,
 \qquad u_{iijk} = - \frac{17}{6} + \frac12 {\rm i} . \]
 In this cube, we now find the integer fields
 \begin{alignat}{6}
 &\Xi_i = -1 , \qquad && \Xi_j = 1 , \qquad&& \Xi_k = 0 , && && &\nonumber\\
 &\Theta_i = -1 , \qquad && \Theta_{ii} = 1 , \qquad && \Theta_{ij} = 0 , \qquad && \Theta_{ki} = 1 ,&& & \nonumber\\
 & && \Theta_{iij} = -1 , \qquad&& \Theta_{ijk} = 0 , \qquad&& \Theta_{kii} = -1 , \qquad&& \Theta_{iijk} = 1 .& \label{intfields2}
 \end{alignat}
 We see that equations \eqref{intfields1} and \eqref{intfields2} give conflicting values for $\Xi_i$, $\Xi_j$, $\Xi_k$, $\Theta_i$, $\Theta_{ki}$, $\Theta_{ijk}$.
\end{ExampleH2}

However, there is no need to restrict the Lagrangian multiform principle to an individual elementary cube. As a general Lagrangian multiform principle in the presence of integer fields, we propose the following definition.

\begin{Definition}
 \label{def-multiform-Gamma}
 $U\colon \Z^N \to \C$ solves the Lagrangian multiform principle for a discrete 2-form~$\cL$, if for every discrete surface $\Gamma \subset \Z^N$, there exist integer fields $\Theta\colon \Gamma \to \Z$ and $\Xi_1, \ldots, \Xi_N \in \Z$ such that the extended action
 \[ S_\Gamma^{\vec \Theta, \vec \Xi} = \sum_{f \text{ face of $\Gamma$}} \cL(f) + \sum_{\vec n \in \Gamma} 2 \pi {\rm i} \Theta(\vec n) U(\vec n) + \sum_{i=1}^N 2 \pi {\rm i} \Xi_i A_i \]
 satisfies
 \begin{equation}
 \label{gamma-partials}
 \pdv{S_\Gamma^{\vec \Theta, \vec \Xi}}{U(\vec n)} = 0
 \end{equation}
 for every lattice site $\vec n$ in the interior of $\Gamma$, and, if $\Gamma$ is a closed surface,
 \begin{equation}
 \label{gamma-closure}
 S_\Gamma^{\vec \Theta, \vec \Xi} = 0 .
 \end{equation}
\end{Definition}

In other words, we allow the integer fields to depend on the surface at hand, but in keeping with to original spirit of Lagrangian multiforms, the same function $U$ of the whole lattice must provide a critical point for all discrete surfaces.

\begin{Proposition}
 $U\colon \Z^N \to \C$ solves the Lagrangian multiform principle if and only if the quad equations are satisfied on every elementary square in $\Z^N$.
\end{Proposition}
\begin{proof}	
 To show that the Lagrangian multiform principle implies all quad equations, it is sufficient to consider surfaces $\Gamma$ that are unit cubes in $\Z^N$.

 Now we show that the Lagrangian multiform principle is satisfied on all solutions of the quad equations. Let $\Gamma$ be an arbitrary discrete surface and $\vec n$ one of its interior vertices. Then there exists a collection of unit cubes such that the faces of $\Gamma$ adjacent to $\vec n$ coincide with the faces of the cubes adjacent to $\vec n$ (where pairs of faces with opposite orientations cancel), see~\mbox{\cite[Lemma~3.3]{boll2014integrability}}. Hence, the corner equations imply \eqref{gamma-partials}.
 Furthermore, if $\Gamma$ is a closed surface, then it can be obtained as the combination of a finite number of unit cubes, where again repeated faces with opposite orientation cancel. The values of $\Xi_i$ that yield \eqref{gamma-closure} can be obtained by summing the corresponding values of each of the elementary cubes.
\end{proof}

To capture our Lagrangian multiforms for H3, Q3 and A2 in Definition \ref{def-multiform-Gamma}, the requirement~\eqref{gamma-closure} needs to be weakened to $S_\Gamma^{\vec \Theta, \vec \Xi} \equiv 0 \mod 4 \pi^2$. It is an open question whether alternative Lagrangian multiforms exists for these equation that satisfy the closure property exactly. Further research into this question, as well as into the closure relation for Q4, should decide which version of the closure condition is most appropriate.

\section{Conclusion}
\label{sec-conclusion}

We have shown that, contrary to common belief, the quad equations of the ABS list are variational. To show this, we introduced a Lagrangian multiform on a four-point stencil instead of the more commonly used triangular stencil. We also took into account the effects of branch cuts in the three-leg forms of the quad equations. To ensure the three-leg forms that occur as corner equations of the Lagrangian multiforms are equivalent to the quad equations, we added an integer field $\Theta$ to the action. It remains to be seen whether this construction can be given a~geometric or topological interpretation.

We pointed out a problem with the usual formulation of the closure property of Lagrangian multiforms for the ABS equations and gaps in the available proofs of this property. We provided a new proof of (almost-)closure for all ABS equations except Q4. Our proof is based on a~continuous deformation of a given solution into a trivial solution and accounting for the effects of crossing branch cuts. We are optimistic that these arguments can be generalised to prove almost-closure of Q4, where we expect the $4 \pi {\rm i}$ will be replaced by a quantity related to the half-periods of the elliptic function underlying this equation. This is left for future work, as it will require non-trivial extensions of Lemma~\ref{lemma-Sjump} and a careful study of the relations between multi-affine variables $v(t)$ and transformed variables $V(t)$ in a suitable limit $t \to 0$.

In part II of this paper, we will study the relation between the triangle and trident Lagrangians and two additional types of Lagrangian multiform, and investigate their double-zero structure.

\appendix

\section{The ABS list}\label{AppendixA}

Here we present the equations of the ABS list. For each one, we give leg functions $\psi$ and $\phi$ making up a three-leg form. There is a large amount of freedom in how to present the equations: variable transformations, adding terms to the leg functions $\psi,\phi$ that cancel against each other, adding constants to the Lagrangian, etc.

For H1--H3, A1--A2 and Q1, the leg functions presented are essentially those from \cite{lobb2009lagrangian}, as well as a more symmetric version of $\Lambda$, and an explicit form for case $\delta = 0$ in H3 and A2. The Lagrangians for the equations Q2--Q3 are inspired on the ideas of that paper and on the three-leg forms presented in \cite{adler2003classification, bobenko2010lagrangian}.

\paragraph{H1.} $Q = (u - u_{ij}) (u_i - u_j) - \alpha_i + \alpha_j$, no transformation required: $U=u$, etc.,
\begin{alignat*}{4}
 &\text{Leg functions:} \qquad && \psi = 2u_i ,\qquad&&
 \phi = \frac{2(\alpha_i - \alpha_j)}{u - u_{ij}} ,& \\
 &\text{Leg Lagrangians:}\qquad&& L = 2 u u_i ,\qquad&&
 \Lambda = 2 (\alpha_i - \alpha_j) \log ( u - u_{ij} ) ,& \\
 &\text{Biquadratics:}\qquad&& h = 1 ,\qquad&&
 g = -\frac{(u-u_{ij})^2}{\alpha_i-\alpha_j} .&
\end{alignat*}
The $\Lambda$ given above does not satisfy the symmetry property \eqref{L-symmetries}. Instead, switching the~$u_i$ and~$u_j$ changes it by $(\alpha_i-\alpha_j) \pi {\rm i}$. An alternative, symmetric, choice of $\Lambda$ is
\[
\Lambda_{{\rm sym}} = (\alpha_i - \alpha_j) \log \bigl( (u - u_{ij})^2 \bigr) .
\]

Note that in the examples throughout the text we scaled $\psi$, $\phi$, $L$, $\Lambda$ by $\frac12$. Those~$L$,~$\Lambda$,~without the prefactor~2, fail to satisfy equation \eqref{L-log(h)} in the proof of Lemma~\ref{lemma-dSdalpha}, but still satisfy the lemma itself. However, if we multiplied $\Lambda_{{\rm sym}}$ by $\frac12$, that result would no longer hold.

\paragraph{H2.} $ Q = (u - u_{ij}) (u_i - u_j) - (\alpha_i - \alpha_j) (u + u_i + u_{ij} + u_j) -\alpha_i^2 + \alpha_j^2 $, no transformation required: $U=u$, etc.,
\begin{alignat*}{3}
 &\text{Leg functions:} \qquad && \psi = \log(\alpha_{i} + u + u_{i}) + 1 ,& \\
 & &&\phi = \log(\alpha_{i} - \alpha_{j} + u - u_{ij}) - \log(-\alpha_{i} + \alpha_{j} + u - u_{ij}) . &\\
 &\text{Leg Lagrangians:} \qquad && L = {(\alpha_{i} + u + u_{i})} \log(\alpha_{i} + u + u_{i}) , &\\
 & &&\Lambda = {(\alpha_{i} - \alpha_{j} + u - u_{ij})} \log(\alpha_{i} - \alpha_{j} + u - u_{ij}) &\\
 & &&\hphantom{\Lambda =}{} + {(\alpha_{i} - \alpha_{j} - u + u_{ij})} \log(-\alpha_{i} + \alpha_{j} + u - u_{ij}) . &\\
 &\text{Biquadratics:} \qquad&& h = \alpha_i + u + u_i , &\\
 & &&g = -\frac{(\alpha_i-\alpha_j + u - u_{ij})(-\alpha_i+\alpha_j + u - u_{ij})}{2(\alpha_i - \alpha_j)} . &
\end{alignat*}
An alternative, symmetric, choice of $\Lambda$ is
\begin{align*}
\Lambda_{{\rm sym}} &= {(\alpha_{i} - \alpha_{j} + u - u_{ij})} \log(\alpha_{i} - \alpha_{j} + u - u_{ij}) \\
&\quad{} + {(\alpha_{i} - \alpha_{j} - u + u_{ij})} \log(\alpha_{i} - \alpha_{j} - u + u_{ij}) + \pi {\rm i} (u + u_{ij}) .
\end{align*}

\paragraph{H3.}
$ Q_{{\rm multi-affine}} = {(u u_{i} + u_{ij} u_{j})} \alpha_{i} - {(u_{i} u_{ij} + u u_{j})} \alpha_{j} + \delta \bigl( \alpha_i^2 - \alpha_j^2 \bigr)$, with transformation \smash{$u = {\rm e}^U$}, \smash{$\alpha_i = -{\rm e}^{A_i}$}, etc.,
\[ Q_{{\rm transformed}} = -\bigl({\rm e}^{U+U_i} + {\rm e}^{U_{ij}+U_j}\bigr) {\rm e}^{A_{i}} + \bigl({\rm e}^{U_{i}+U_{ij}} + {\rm e}^{U+U_{j}}\bigr) {\rm e}^{A_{j}} + \delta \bigl( {\rm e}^{2A_i} - {\rm e}^{2A_j} \bigr), \]
\begin{alignat*}{3}
& \text{Leg functions:}\qquad&& \psi =
\begin{cases}
 -U_i & \text{if }\delta = 0 , \\
 -A_i - \log \bigl( 1 - {\rm e}^{U+U_i-A_i} \bigr) & \text{if }\delta = 1 ,
\end{cases} &\\
& && \phi = \log \bigl( 1- {\rm e}^{A_i-A_j+U-U_{ij}} \bigr) - \log \bigl(1 -{\rm e}^{-A_i+A_{j}+U-U_{ij}} \bigr) - A_i + A_j , & \\
& \text{Lagrangians:}\qquad && L =
\begin{cases}
 -U U_i& \text{if }\delta = 0 , \\
 \dilog({\rm e}^{-A_i+U+U_i}) -A_i (U + U_i) & \text{if }\delta = 1 ,
\end{cases} &\\
& && \Lambda = \dilog \bigl( {\rm e}^{-A_i+A_{j}+U-U_{ij}} \bigr) - \dilog \bigl( {\rm e}^{A_i-A_{j}+U-U_{ij}} \bigr) - (A_i-A_j) (U-U_{ij}) ,& \\
& \text{Biquadratics:}\qquad&& h = u u_i + \delta \alpha_i = {\rm e}^{U + U_i} - \delta {\rm e}^{A_i} , &\\
& && g = \frac{(\alpha_i u - \alpha_j u_{ij})(\alpha_j u - \alpha_i u_{ij})}{\alpha_i^2 - \alpha_j^2}& \\
& && \hphantom{g}{} = \frac{{\rm e}^{2 U_{ij}}}{{\rm e}^{A_i - A_j} - {\rm e}^{A_j - A_i}} \bigl(1 - {\rm e}^{A_i - A_j + U - U_{ij}} \bigr) \bigl(1 - {\rm e}^{-A_i + A_j + U - U_{ij}} \bigr) .&
\end{alignat*}
An alternative, symmetric, choice of $\Lambda$ is
\[
\Lambda_{{\rm sym}} = -\dilog \bigl( {\rm e}^{A_i-A_{j}+U-U_{ij}} \bigr) - \dilog \bigl( {\rm e}^{A_i-A_{j}-U+U_{ij}} \bigr) - \frac12 (U - U_{ij})^2 + (U+U_{ij}) \pi {\rm i} .
\]

\paragraph{A1$\boldsymbol{_{\delta=0}}$.} $Q = {(u u_{ij} + u_{i} u_{j})} {(\alpha_{i} - \alpha_{j})} + {(u u_{i} + u_{ij} u_{j})} \alpha_{i} - {(u_{i} u_{ij} + u u_{j})} \alpha_{j}$, no transformation required: $U=u$, etc.,
\begin{alignat*}{4}
 &\text{Leg functions:}\qquad&& \psi = \frac{2 \alpha_{i}}{u + u_{i}} ,\qquad&&
\phi = \frac{2 (\alpha_{i} - \alpha_{j})}{u - u_{ij}} ,& \\
 &\text{Lagrangians:}\qquad && L = 2 \alpha_{i} \log(u + u_{i}) ,\qquad&&
\Lambda = 2 (\alpha_{i} - \alpha_{j}) \log (u - u_{ij}) ,& \\
& && && \Lambda_{{\rm sym}} = (\alpha_{i} - \alpha_{j}) \log \bigl((u - u_{ij})^2\bigr) ,& \\
& \text{Biquadratics:}\qquad && h = \frac{(u + u_i)^2}{\alpha_i} ,\qquad&&
g = \frac{(u - u_{ij})^2}{\alpha_i - \alpha_j} .&
\end{alignat*}
As in the case of H1, we could scale $\psi$, $\phi$, $L$, $\Lambda$ by $\frac12$ (but not $\Lambda_{{\rm sym}}$). The rescaled $L$, $\Lambda$ would not satisfy equation \eqref{L-log(h)}, but Lemma~\ref{lemma-dSdalpha} would still hold.

\paragraph{A1$\boldsymbol{_{\delta=1}}$.} $Q = {(u u_{ij} + u_{i} u_{j})} {(\alpha_{i} - \alpha_{j})} + {(u u_{i} + u_{ij} u_{j})} \alpha_{i} - {(u_{i} u_{ij} + u u_{j})} \alpha_{j} - {(\alpha_{i} - \alpha_{j})} \alpha_{i} \alpha_{j}$, no transformation required: $U=u$, etc.,
\begin{alignat*}{3}
 &\text{Leg functions:}\qquad&& \psi = \log(\alpha_{i} + u + u_{i}) - \log(-\alpha_{i} + u + u_{i}) ,& \\
 & &&\phi = \log(-\alpha_{i} + \alpha_{j} - u + u_{ij}) - \log(\alpha_{i} - \alpha_{j} - u + u_{ij}) , &\\
 &\text{Lagrangians:}\qquad&& L = {(\alpha_{i} + u + u_{i})} \log(\alpha_{i} + u + u_{i}) + {(\alpha_{i} - u - u_{i})} \log(-\alpha_{i} + u + u_{i}) ,& \\
 & &&\Lambda = {(\alpha_{i} - \alpha_{j} - u + u_{ij})} \log(\alpha_{i} - \alpha_{j} - u + u_{ij})& \\
 & &&\hphantom{\Lambda =}{} + {(\alpha_{i} - \alpha_{j} + u - u_{ij})} \log(-\alpha_{i} + \alpha_{j} - u + u_{ij}) ,& \\
 & &&\Lambda_{{\rm sym}} = {(\alpha_{i} - \alpha_{j} - u + u_{ij})} \log(\alpha_{i} - \alpha_{j} - u + u_{ij})& \\
 & &&\hphantom{\Lambda_{{\rm sym}} =}{} + {(\alpha_{i} - \alpha_{j} + u - u_{ij})} \log(\alpha_{i} - \alpha_{j} + u - u_{ij}) + (u+u_{ij}) \pi {\rm i} ,& \\
 &\text{Biquadratics:}\qquad&& h = \frac{(\alpha_i + u + u_i)(-\alpha_i + u + u_i)}{\alpha_i} ,& \\
 & &&g = \frac{(\alpha_i-\alpha_j + u - u_{ij})(-\alpha_{i} + \alpha_{j} + u - u_{ij})}{\alpha_i - \alpha_j}.&
\end{alignat*}

\paragraph{A2.}
$Q_{{\rm multi-affine}} = (u_{i} u_{ij} + u u_{j}) \bigl(\alpha_i- \frac{1}{\alpha_i} \bigr) - (u u_{i} + u_{ij} u_{j}) \bigl(\alpha_j - \frac{1}{\alpha_j} \bigr) - (u u_{i} u_{ij} u_{j} + 1) \bigl( \frac{\alpha_i}{\alpha_j} - \frac{\alpha_j}{\alpha_i} \bigr) $, with transformation $u = {\rm e}^U$, $\alpha_i = {\rm e}^{A_i}$, etc.,
\begin{align*}
Q_{{\rm transformed}} &= \bigl({\rm e}^{U + U_j} + {\rm e}^{U_i + U_{ij}}\bigr) \sinh(A_i) - \bigl({\rm e}^{U+U_i} + {\rm e}^{U_{ij}+U_j}\bigr) \sinh(A_j) \\
&\quad{} -\bigl({\rm e}^{U+U_i+U_{ij}+U_j} + 1\bigr) \sinh(A_i - A_j) ,
\end{align*}
\unskip
\begin{alignat*}{3}
 &\text{Leg functions:}\qquad&& \psi = \log\bigl(1-{\rm e}^{A_i+U+U_i}\bigr) - \log\bigl(1-{\rm e}^{-A_i+U+U_i}\bigr) - A_i ,& \\
 & && \phi = \log \bigl( 1-{\rm e}^{A_i-A_j+U-U_{ij}} \bigr) - \log \bigl( 1-{\rm e}^{-A_i+A_j+U-U_{ij}} \bigr) - A_i + A_j , &\\
 &\text{Lagrangians:}\qquad && L = \dilog\bigl({\rm e}^{-A_i + U + U_i}\bigr) - \dilog\bigl({\rm e}^{A_i + U + U_i}\bigr) - A_i(U+U_i) ,& \\
& && \Lambda = \dilog\bigl({\rm e}^{-A_i+A_j+U-U_{ij}}\bigr) - \dilog\bigl({\rm e}^{A_i-A_j+U-U_{ij}}\bigr) - (A_i - A_j) (U - U_{ij}) ,& \\
& && \Lambda_{{\rm sym}} = -\dilog\bigl({\rm e}^{A_i-A_j+U-U_{ij}}\bigr) - \dilog\bigl({\rm e}^{A_i-A_j-U+U_{ij}}\bigr) \\
& && \hphantom{\Lambda_{{\rm sym}} =}{}- \tfrac12 (U - U_{ij})^2 + (U+U_{ij}) \pi {\rm i} ,& \\
& \text{Biquadratics:}\qquad&& h = \frac{ ( 1 - \alpha_i u u_i) \bigl(1 - \frac{u u_i}{\alpha_i} \bigr)}{ \alpha_i - \frac{1}{\alpha_i}}
 = \frac{ \bigl(1- {\rm e}^{A_i + U + U_i}\bigr) \bigl(1 - {\rm e}^{-A_i + U + U_i}\bigr)}{ 2 \sinh(A_i- A_j)} , &\\
& && g = u_{ij}^2 \frac{ \bigl( 1- \frac{\alpha_i u}{\alpha_j u_{ij}} \bigr) \bigl( 1- \frac{\alpha_j u}{\alpha_i u_{ij}} \bigr)}{ \frac{\alpha_i}{\alpha_j} - \frac{\alpha_j}{\alpha_i}}& \\
& && \hphantom{g}{} = {\rm e}^{2 U_{ij}} \frac{ \bigl(1 - {\rm e}^{A_i - A_j + U - U_{ij}} \bigr) \bigl( 1-{\rm e}^{-A_i + A_j + U - U_{ij}} \bigr)}{ 2 \sinh(A_i- A_j)} .&
\end{alignat*}

\paragraph{Q1$\boldsymbol{_{\delta = 0}}$.}
$ Q = \alpha_{i} {(u - u_{j})} {(u_{i} - u_{ij})} + \alpha_{j} {(u - u_{i})} {(u_{ij} - u_{j})} $, no transformation required: $U=u$, etc.,
\begin{alignat*}{4}
 &\text{Leg functions:}\qquad && \psi = \frac{2\alpha_{i}}{u - u_{i}} ,\qquad&&
 \phi = \frac{2(\alpha_i - \alpha_j)}{u - u_{i}} ,& \\
 &\text{Lagrangians:}\qquad && L = 2 \alpha_{i} \log( {u - u_{i}} ) ,\qquad&&
 \Lambda = 2 (\alpha_i-\alpha_j) \log( {u - u_{ij}} ) , &\\
 & && L_{{\rm sym}} = \alpha_{i} \log \bigl( ({u - u_{i}})^2 \bigr) ,\qquad&& \Lambda_{{\rm sym}} = (\alpha_i-\alpha_j) \log \bigl( (u - u_{ij})^2 \bigr) ,& \\
 &\text{Biquadratics:}\qquad&& h = \frac{(u-u_i)^2}{\alpha_i} ,\qquad&&
 g = -\frac{(u-u_{ij})^2}{\alpha_i - \alpha_j} .&
\end{alignat*}
For equations of type Q, the long leg function $\phi$ is identical to the short leg function $\psi$, and correspondingly $\Lambda$ is identical to $L$. The biquadratics $g$ and $h$ are the same up to an overall minus sign. For the remaining equations we will only write $\psi$, $L$ and $h$.

As in the case of H1, we could scale $\psi$, $\phi$, $L$, $\Lambda$ by $\frac12$ (but not $\Lambda_{{\rm sym}}$). The rescaled $L$, $\Lambda$ would not satisfy equation \eqref{L-log(h)}, but Lemma~\ref{lemma-dSdalpha} would still hold.

\paragraph{Q1$\boldsymbol{_{\delta=1}}$.}
$Q = {(\alpha_{i} - \alpha_{j})} \alpha_{i} \alpha_{j} + \alpha_{i} {(u - u_{j})} {(u_{i} - u_{ij})} + \alpha_{j} {(u - u_{i})} {(u_{ij} - u_{j})} $, no transformation required: $U=u$, etc.,
\begin{align*}
&\psi= \log(\alpha_{i} + u - u_{i}) - \log(-\alpha_{i} + u - u_{i}) , \\
&L= (\alpha_{i} + u - u_{i}) \log(\alpha_{i} + u - u_{i}) + (\alpha_{i} - u + u_{i}) \log(-\alpha_{i} + u - u_{i}) , \\
&L_{{\rm sym}}= (\alpha_{i} + u - u_{i}) \log(\alpha_{i} + u - u_{i}) + (\alpha_{i} - u + u_{i}) \log(\alpha_{i} - u + u_{i}) + (u+u_i) \pi {\rm i} , \\
&h= -\frac{(\alpha_i + u-u_i)(-\alpha_i + u-u_i)}{\alpha_i} .
\end{align*}

\paragraph{Q2.}
$Q_{{\rm multi-affine}} = -{(\alpha_{i}^{2} - \alpha_{i} \alpha_{j} + \alpha_{j}^{2} - u - u_{i} - u_{ij} - u_{j})} {(\alpha_{i} - \alpha_{j})} \alpha_{i} \alpha_{j} + \alpha_i (u - u_j)(u_i - u_{ij}) - \alpha_j (u - u_i)(u_j - u_{ij})$, with transformation $u = U^2$, $\alpha_i = A_i$, etc.,
\begin{gather*}
Q_{{\rm transformed}} = -\bigl(A_i^2 - A_i A_j + A_j^2 - U^2 - U_i^2 - U_{ij}^2 - U_j^2\bigr) (A_i - A_j) A_i A_j \\
\hphantom{Q_{{\rm transformed}} =}{} + A_i \bigl(U^2 - U_j^2\bigr) \bigl(U_i^2 - U_{ij}^2\bigr) + A_j \bigl(U^2 - U_i^2\bigr) \bigl(U_{ij}^2 - U_j^2\bigr) ,
\\
\psi = \log(A_i + U + U_i) + \log(A_i + U - U_i) - \log(-A_i + U - U_i) - \log(-A_i + U + U_i) , \\
L = (A_i + U + U_i) \log(A_i + U + U_i) + (A_i + U - U_i) \log(A_i + U - U_i) \\
\hphantom{L =}{} + (A_i - U - U_i) \log(-A_i + U + U_i) + (A_i - U + U_i) \log(-A_i + U - U_i) , \\
L_{{\rm sym}} = (A_i + U + U_i) \log(A_i + U + U_i) + (A_i + U - U_i) \log(A_i + U - U_i) \\
\hphantom{L_{{\rm sym}} =}{} + (A_i - U - U_i) \log(-A_i + U + U_i) + (A_i - U + U_i) \log(A_i - U + U_i) \\
\hphantom{L_{{\rm sym}} =}{} + (U + U_i) \pi {\rm i} , \\
h = \frac{1}{\alpha_i}\bigl( \alpha_i^4 - 2 \alpha_i^2 (u+u_i) + (u-u_i)^2 \bigr) \\
\hphantom{h }{} = \frac{1}{A_i} (A_i + U + U_i) (A_i + U - U_i) (-A_i + U - U_i) (-A_i + U + U_i) .
\end{gather*}

\paragraph{Q3$\boldsymbol{_{\delta=0}}$.}
$Q_{{\rm multi-affine}} = {(u u_{i} + u_{ij} u_{j})} \bigl(\alpha_i- \frac{1}{\alpha_i} \bigr) - {(u_{i} u_{ij} + u u_{j})} \bigl(\alpha_j - \frac{1}{\alpha_j} \bigr) - {(u u_{ij} + u_{i} u_{j})} \bigl( \frac{\alpha_i}{\alpha_j} - \frac{\alpha_j}{\alpha_i} \bigr)$, with transformation $u = {\rm e}^U$, $\alpha_i = {\rm e}^{A_i}$, etc.,
\begin{gather*}
Q_{{\rm transformed}} = \bigl({\rm e}^{U+U_i} + {\rm e}^{U_{ij}+U_j}\bigr) \sinh(A_i) - \bigl({\rm e}^{U_i+U_{ij}} + {\rm e}^{U+U_j}\bigr) \sinh(A_j) \\
 \hphantom{Q_{{\rm transformed}} =}{} - \bigl({\rm e}^{U+U_{ij}} + {\rm e}^{U_i+U_j}\bigr) \sinh(A_i - A_j), \\
\psi = \log \bigl(1-{\rm e}^{A_i + U - U_i}\bigr) - \log \bigl(1-{\rm e}^{-A_i + U - U_i}\bigr) - A_i , \\
L = \dilog \bigl({\rm e}^{-A_i + U - U_i}\bigr) - \dilog \bigl({\rm e}^{A_i + U - U_i}\bigr)
- A_i (U - U_i) , \\
L_{{\rm sym}} = -\dilog \bigl({\rm e}^{A_i + U - U_i}\bigr)
- \dilog \bigl({\rm e}^{A_i - U + U_i}\bigr)
- \frac12 (U - U_i)^2 + (U + U_i) \pi {\rm i} , \\
h = \frac{ (u_i - \alpha_i u ) (\alpha_i u_i - u)}{ \alpha_i^2 - 1}
= \frac{{\rm e}^{2 U_i} \bigl( 1-{\rm e}^{A_i + U - U_i} \bigr) \bigl( 1-{\rm e}^{-A_i + U - U_i} \bigr)}{ 2 \sinh(A_i)} .
\end{gather*}

\paragraph{Q3$\boldsymbol{_{\delta=1}}$.}
$Q_{{\rm multi-affine}} = \bigl(\alpha_i- \frac{1}{\alpha_i} \bigr) (u u_i + u_{ij} u_j) - \bigl(\alpha_j- \frac{1}{\alpha_j} \bigr) (u_i u_{ij} + u u_j) - \bigl( \frac{\alpha_i}{\alpha_j} - \frac{\alpha_j}{\alpha_i} \bigr) (u u_{ij} + u_i u_j)
- \bigl( \frac{\alpha_i}{\alpha_j} - \frac{\alpha_j}{\alpha_i} \bigr) \bigl(\alpha_i- \frac{1}{\alpha_i} \bigr) \bigl(\alpha_j- \frac{1}{\alpha_j} \bigr) $, with transformation $u = 2\cosh(U)$, $\alpha_i = {\rm e}^{A_i}$, etc.,
\begin{gather*}
Q_{{\rm transformed}} = \sinh(A_{i}) {(\cosh(U) \cosh(U_{i}) + \cosh(U_{ij}) \cosh(U_{j}))} \\
 \hphantom{Q_{{\rm transformed}} =}{} - \sinh(A_{j}) {(\cosh(U_{i}) \cosh(U_{ij}) + \cosh(U) \cosh(U_{j}))} \\
 \hphantom{Q_{{\rm transformed}} =}{} - \sinh(A_{i} - A_{j})(\cosh(U) \cosh(U_{ij}) + \cosh(U_{i}) \cosh(U_{j}) ) \\
 \hphantom{Q_{{\rm transformed}} =}{} - \sinh(A_{i} - A_{j})\sinh(A_{i}) \sinh(A_{j}),
\\
\psi = \log \bigl( 1-{\rm e}^{-A_i - U + U_i} \bigr) + \log \bigl( 1-{\rm e}^{-A_i - U - U_i} \bigr) \\
 \hphantom{\psi =}{} - \log\bigl( 1-{\rm e}^{A_i - U - U_i} \bigr) - \log \bigl( 1-{\rm e}^{A_i - U + U_i}\bigr) + A_i , \\
L = \dilog \bigl( {\rm e}^{-A_i - U + U_i} \bigr) + \dilog \bigl( {\rm e}^{-A_i - U - U_i}\bigr) -\dilog \bigl( {\rm e}^{A_i - U - U_i} \bigr) - \dilog \bigl( {\rm e}^{A_i - U + U_i} \bigr) \\
 \hphantom{L =}{} + A_i (U - U_i) , \\
L_{{\rm sym}} = \dilog \bigl( {\rm e}^{-A_i + U + U_i} \bigr) + \dilog\bigl( {\rm e}^{-A_i + U - U_i} \bigr) + \dilog \bigl( {\rm e}^{-A_i - U + U_i} \bigr) + \dilog \bigl( {\rm e}^{-A_i - U - U_i} \bigr) \\
 \hphantom{L_{{\rm sym}} =}{} + U^2 + U_i^2 - A_i (U - U_i) , \\
h = \frac{\alpha_i^2 - 2 + \frac{1}{\alpha_i^2} - \bigl(\alpha_i + \frac{1}{\alpha_i}\bigr) u u_i + u^2 + u_i^2}{ \alpha_i - \frac{1}{\alpha_i}} \\
\hphantom{h}{} = \frac{2 \sinh(A_i)^2 - 4 \cosh(A_i) \cosh(U) \cosh(U_i) + 2 \cosh(U)^2 + 2\cosh(U_i)^2}{ \sinh(A_i) } \\
\hphantom{h}{} = \frac{{\rm e}^{2 U} \bigl(1 - {\rm e}^{-A_i-U+U_i}\bigr) \bigl(1 - {\rm e}^{-A_i-U-U_i}\bigr) \bigl(1 - {\rm e}^{A_i-U-U_i}\bigr) \bigl(1 - {\rm e}^{A_i-U+U_i}\bigr)}{ {\rm e}^{A_i} - {\rm e}^{-A_i} }.
\end{gather*}

\subsection*{Acknowledgements}

The impetus for this work was provided by the critical questions asked by an anonymous referee of the paper that has now become Part II of the present work \cite{richardson2024discrete2}. We are grateful for their detailed feedback and constructive criticism.
We would like to thank Prof Frank Nijhoff and Dr Vincent Caudrelier for helpful discussions on the topic of this work and many related subjects.

JR acknowledges funding from Engineering and Physical Sciences Research Council DTP, Crowther Endowment and School of Mathematics at the University of Leeds. MV is supported by the Engineering and Physical Sciences Research Council [EP/Y006712/1].

\pdfbookmark[1]{References}{ref}
\LastPageEnding

\end{document}